\documentclass[a4paper,11pt]{article}
\usepackage[T1]{fontenc}
\usepackage[utf8]{inputenc}
\usepackage{tabularx}
\usepackage{supertabular}
\usepackage{booktabs}
\RequirePackage[authoryear]{natbib}
\usepackage{lmodern}
\usepackage{setspace}
\usepackage{nicefrac}
\usepackage{amsmath,dsfont}
\usepackage{framed}
\usepackage{a4wide}
\usepackage{amssymb,amsmath,amsthm}
\usepackage{bbm}
\usepackage{changes}
\usepackage[flushleft]{threeparttable}
\usepackage{float}
\newcommand{\RNum}[1]{\uppercase\expandafter{\romannumeral #1\relax}}
\newcommand{\RomanNumeralCaps}[1]
    {\MakeUppercase{\romannumeral #1}}
\newcommand{\Pa}{\mathbb{P}_{\hspace*{-.1em}\alpha}}
\newcommand{\Paa}{\mathbb{P}_{\hspace*{-.1em}\tilde\alpha}}
\definecolor{dblue}{rgb}{0.21,0.21,0.55}
\definecolor{light-blue}{rgb}{0.63, 0.79, 0.95}
\newcommand{\dbl}{\color{dblue}}
\definecolor{darkbrown}{rgb}{0.8, 0.58, 0.46}

\usepackage{multirow}
\usepackage{lscape}
\usepackage[english]{babel}
\usepackage{graphicx}
\usepackage{arydshln}
\usepackage{rotating}
\usepackage[flushmargin,hang]{footmisc}
\usepackage{amsmath}
\usepackage{lscape}
\usepackage[pdflinkmargin=5pt,pdfstartview={FitBH -32768},plainpages=false]{hyperref}

\renewcommand{\P}{\mathbb{P}}

\newcommand{\E}{\mathbb{E}}
\newcommand{\N}{\mathbb{N}}
\newcommand{\R}{\mathbb{R}}
\newcommand{\1}{\mathbbm{1}}
\newcommand{\todl}{\stackrel{d}{\longrightarrow}}
\newcommand{\KLEINO}{{\scriptstyle{\mathcal{O}}}}
\DeclareMathAccent{\verywidehat}{\mathord}{largesymbols}{'144}

\newcommand{\cov}{\mathbb{C}\textnormal{o\hspace*{0.02cm}v}}

\allowdisplaybreaks[4]

\newtheorem{theo}{Theorem}
\newtheorem{assump}{Assumption}
\newtheorem{prop}{Proposition}[section]
\newtheorem{lem}{Lemma}

\begin{document}

\title{Probabilistic models and statistics for electronic financial markets in the digital age.}
\author{Markus Bibinger\footnote{Financial support from the Deutsche Forschungsgemeinschaft (DFG) under grant 403176476 is gratefully acknowledged.}\\
\emph{\small Faculty of Mathematics and Computer Science,}\\ \emph{\small Institute of Mathematics},\\ \emph{\small University of W\"urzburg}}
\date{Draft \today}
%
\maketitle\thispagestyle{empty}
\begin{abstract}
{{\normalsize \noindent 
The scope of this manuscript is to review some recent developments in statistics for discretely observed semimartingales which are motivated by applications for financial markets. Our journey through this area stops to take closer looks at a few selected topics discussing recent literature. We moreover highlight and explain the important role played by some classical concepts of probability and statistics. We focus on three main aspects: Testing for jumps; rough fractional stochastic volatility; and limit order microstructure noise. We review jump tests based on extreme value theory and complement the literature proposing new statistical methods. They are based on asymptotic theory of order statistics and the R\'{e}nyi representation. The second stage of our journey visits a recent strand of research showing that volatility is rough. We further investigate this and establish a minimax lower bound exploring frontiers to what extent the regularity of latent volatility can be recovered in a more general framework. Finally, we discuss a stochastic boundary model with one-sided microstructure noise for high-frequency limit order prices and its probabilistic and statistical foundation. }}
\vspace{.25cm} 

%
\end{abstract}
\section{Introduction}

The evolutions of stock prices are subject to market risk. Foundations of price models in continuous time typically refer back to Louis Bachelier. He did his PhD supervised by Henri Poincar\'{e} in Paris and defended his thesis ``Th\'{e}orie de la sp\'{e}culation'' in 1900. He is considered to be the first researcher who found the \emph{Brownian motion}. I would see this as a great success, although the moral that \emph{price changes cannot be forecasted} (Efficient Market Hypothesis) is rather negative for speculators. Naturally, forecasting of future prices is a worthwhile pursuit. The first Bachelor student I have supervised asked me about methods to do that, but I rather pointed him at \emph{risk forecasting}. Looking at data from the DAX he concluded that ``One can see clearly that there is a much higher autocorrelation for squared returns than for returns. This is an indication that it might be easier to forecast variance of returns than the returns themselves''. 
While forecasting price changes (=returns) would be desirable to make money, forecasting risk is more successful and an integral instrument of risk management. This is a main application of statistics for financial markets and the analysis of financial time series. 

The Brownian motion, or \emph{Wiener process}, $(W_t)$ with $W_0=0$ is defined by the properties:
\begin{itemize}
\setlength{\itemsep}{.05cm}
\item Its increments $(W_{t_2}-W_{t_1},W_{t_3}-W_{t_2},\ldots,W_{t_{n}}-W_{t_{n-1}})$
are independent for all $0< t_1\le t_2\le \cdots\le t_n<\infty$;
\item Its increments are stationary, such that $(W_{t+s}-W_t)$ is distributed as $W_s$;
\item The expectation is $\E[W_t]=0$, for all $t$;
\item The \emph{paths} $t\mapsto W_t$ (=realizations) are continuous.
\end{itemize}
All random objects throughout the manuscript are defined on some probability space $(\Omega,\mathcal{F},\P)$, with $\sigma$-field $\mathcal{F}$ and measure $\P$. We follow the standard notation not to write arguments $\omega\in\Omega$ for random objects.
Brownian motion can be motivated as integrated continuous-time white noise, and as well as limit of a discrete-time \emph{random walk} $X_T=\sum_{t=1}^T\epsilon_t$, $T\in\N$, where $(\epsilon_t)$ are independent, identically distributed (i.i.d.) with $\P(\epsilon_t=1)=1/2=\P(\epsilon_t=-1)$. It holds that $T^{-1/2}X_{\lfloor T s\rfloor}\rightarrow W_s$, as $T\to\infty$, with the floor function ${\lfloor \,\cdot\,\rfloor}$. Note that continuous-time white noise does not exist in the sense of a measurable stochastic process related to the fact, that the paths of $(W_t)$ are continuous, but nowhere differentiable. So each realization of a Brownian motion has this fascinating property like the Weierstrass function. The existence was proved by Wiener in 1923 and we refer to the textbook \cite{schilling2014brownian} for an overview on different constructions and various properties. The Brownian motion is really at the heart of the theory on continuous-time stochastic processes. From my point of view, stochastic processes is mainly the study of classes of processes which share one or some of the fundamental properties of the Brownian motion:
\begin{itemize}
\setlength{\itemsep}{.05cm}
\item It is a \emph{Gaussian process}. That means that all finite-dimensional distributions of $W_{t_1},\ldots,W_{t_n}$ are normal for all $n\in\N$, and arbitrary times. A Gaussian process is uniquely determined by its expectation and covariance function. A Brownian motion is hence uniquely characterized as a continuous Gaussian process with $W_0 = 0$, $\E[W_t]=0$, for all $t$, and the covariance function \(\cov(W_s,W_t)=\min(s,t)\).
\item It is a \emph{L\'{e}vy process}, that is a process $(X_t)$ with independent stationary increments, $X_0=0$, and which satisfies $\forall\epsilon>0$: $\P(|X_{t+h}-X_t|>\epsilon)\to 0$, as $h\to 0$. Writing $X_t=\sum_{j=1}^n(X_{t_j}-X_{t_{j-1}}), 0=t_0< t_1\le t_2\le \cdots\le t_n=t$, their study is related to studying sums of i.i.d.\ random variables. The second fundamental example of a L\'{e}vy process is a Poisson (jump) process. 
\item It is a \emph{martingale} whose conditional expectation satisfies $\E[W_t|W_s]=W_s$, almost surely for all $t\ge s$. 
\item It is a \emph{Markov process}, for which $(W_{t+s}-W_t)_{s\ge 0}$ is another Brownian motion independent of $\{W_u,0\le u\le t\}$.
\item It is \emph{self-similar} such that $a^{-1/2}W_{at}$ is distributed as $W_t$, for all $a>0$. Looking at a path in a plot with axes that have no labels, we could hence not say anything about the scaling.
\end{itemize}
 
The famous {\emph{Black-Scholes price model}} follows Bachelier's principles and describes the stock price $S_t$ at time $t$ by the stochastic differential equation
\begin{align}\label{BS}\text{d}S_t=a S_t\,\text{d}t+\sigma S_t \text{d}W_t\,\end{align} 
which is solved by a geometric Brownian motion. This can be proved by a simple application of It\^{o}'s lemma. Bachelier's mantra, that there is no (short-term) expected profit for traders without insider information, is culminated in the ``Fundamental Theorem of Asset Pricing'': In an arbitrage-free market, prices follow martingales. An extension to general \emph{no arbitrage conditions} by \cite{delbaen} implies that log-prices should be {\emph{semimartingales}}. These are processes that can be expressed on compact time intervals as sums of a martingale and a process of finite variation. Of course, the Brownian motion is a semimartingale. Interestingly, these processes do not only occur in this modern fundamental theorem of asset pricing, but are also the class of ``good integrators'' for stochastic integrals according to the Bichteler-Dellacherie theorem. We work with a continuous-time log-price modelled as an \emph{It\^{o} semimartingale}:
\begin{align}\label{SM}X_t=C_t+\;\text{JUMPS},~\text{with}~ C_t=X_0+\int_0^t \mu_s\,\text{d}s+\int_0^t\sigma_s\,\text{d}W_s,t\in[0,1]\,.\end{align}
Throughout this work we focus on real-valued, one-dimensional processes modelling the price evolution of only one asset. Since the multivariate setting has some surprises in store, a visit is as well interesting. Most of my own work is in fact devoted to multivariate phenomena and we do not live in a one-dimensional world. However, for my selection of topics in this manuscript and simplicity it is sufficient to stick to a one-dimensional image space at a fixed time. Some recent aspects of a multi-dimensional analysis and its applications to portfolios are briefly mentioned in the outlook. In \eqref{SM}, $(\sigma_s)$ is the \emph{volatility} process. Volatility is the prevalent concept to describe market risk. It will therefore be our main target of statistical inference. The first important advancement compared to the Black-Scholes model \eqref{BS} is to include time-varying volatility which may be a stochastic process itself. A second important advancement is to consider \emph{price jumps} which can describe rapid price adjustments in response to new information provided, e.g., by economic shocks or central bank announcements.

The realm of big data in financial markets can be viewed as a stroke of luck for statisticians and data analysts giving us huge data sets at hand, in particular when looking at intra-daily tick data. In electronic financial markets almost 70\% of the volume is nowadays attributed to high-frequency trading. This should be very useful for risk quantification, but similar as in other fields, the picture how to efficiently exploit big data is not yet complete as the data sets are complex and noisy. Nevertheless, or maybe for this very reason, for applications to intra-day financial data in econometrics and also in macroeconomic studies high-frequency statistics for semimartingales became highly relevant. The use of high-frequency financial data was promoted by \cite{engle} and \cite{andersen}, among others, around the turn of the millennium. \cite{engle} called the situation when prices from all transactions are recorded ``ultra-high frequency data''. Such data samples are rather small compared to the ones we now have available from limit order books, see Figure \ref{Fig:LOMN} for a snapshot. The advent of ultra-high frequency data motivated to bridge several strands of research between statistics for stochastic processes, time series analysis and econometrics. Since then high-frequency statistics for semimartingales has evolved into a huge field of study. Many brilliant researchers made substantial contributions to this field including Yacine A\"{i}t-Sahalia, Ole Barndorff-Nielsen, Jean Jacod, Per Mykland and younger ones as Mark Podolskij and Viktor Todorov, to  name just a few. 

Having self-similar processes in price models, it might be a bit disappointing to learn that nevertheless different models are used when considering data over different time scales. While for instance, under a very high time resolution over a short interval discreteness of prices in the image space becomes relevant, what is not in line with a Brownian motion, looking at low time resolutions, e.g., daily data over a year, discrete time series might be adequate models. Depending on the time resolution and the concrete application, micro-, meso- or macroscopic models are used as devices to coherently describe price dynamics and to allow at the same time a good calibration of the model. Often the applied focus is on risk forecasting, portfolio allocation or asset pricing. Since the seminal works on AutoRegressive Conditional Heteroskedastic (ARCH) time series and its generalizations (GARCH), see \cite{engle1986},\,\footnote{Robert F.\ Engle received the Nobel prize in 2003 with Granger. Scholes and Merton won it in 1997 when Black's key role was pointed out, but the prize is not given posthumously.} in the late 1980s, time series models are the main workhorse to perform volatility forecasting, typically on a daily basis. GARCH models take into account several stylized facts as volatility clustering. In the era of intra-day high-frequency price recordings, one common approach is to infer volatility based on a continuous-time model and to plug the estimates into time series models used over daily time scales, see, e.g., \cite{hansen2012realized}. Although it might appear odd to a mathematician, that the continuous-time model is then not used over all time scales, many econometricians are satisfied with the good empirical performance. In recent years the research on forecasting shifted more towards fractional time series and continuous-time fractional models. The rough volatility literature and the data example in Section \ref{sec:4} inherit a similar philosophy to combine estimates from different models for different time resolutions.

Stylized facts of ultra high-frequency data contradict a pure semimartingale model due to so-called \emph{market microstructure}. An early paper reporting these empirical facts and which foreshadowed the very successful semimartingale with additive noise model to describe tick data was \cite{zhou}. Various related estimation methods for the volatility in this model have been proposed including two-scale and multi-scale realized volatility by \cite{zhangmykland} and \cite{zhang}, pre-averaging by \cite{JLMPV}, the kernel estimator by \cite{bn2}, the Quasi-Maximum-Likelihood approach by \cite{xiu} and many more. A lower bound for the asymptotic variance of integrated volatility estimators at optimal rate for this problem was established by \cite{reiss}. In a related vein, the presented model for limit order book quotes in Section \ref{sec:5} preserves the idea of an underlying semimartingale efficient price and describes market microstructure effects by additive noise. The only difference, which is however crucial, is that we proceed from regular noise with expectation zero to irregular, non-negative one-sided noise.
\begin{figure}[t]
\includegraphics[width=6.8cm]{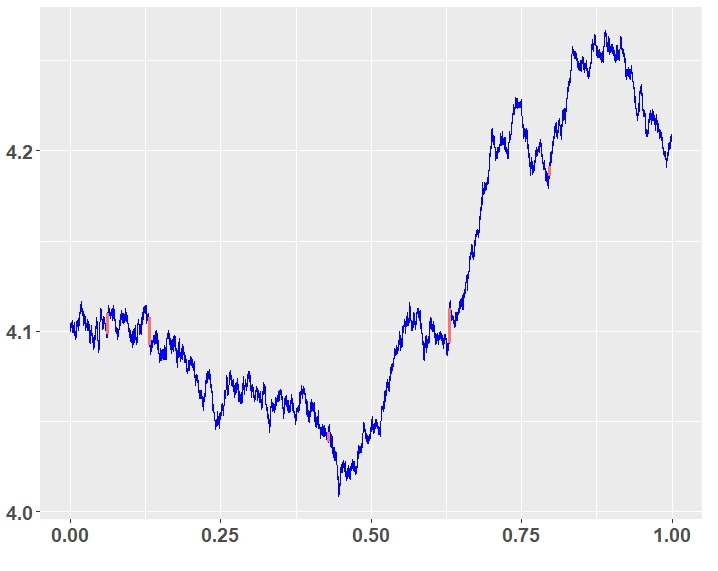}\includegraphics[width=7.2cm]{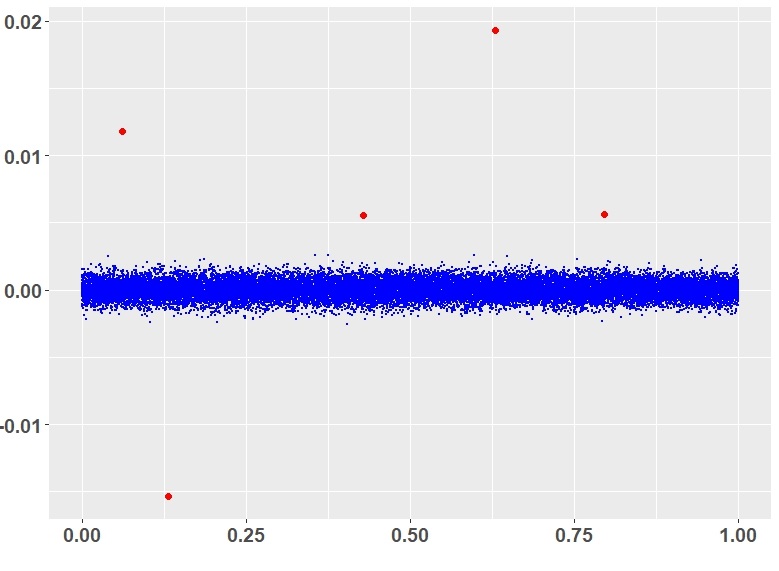}
\caption{\label{Fig:jumps}Log-price (left) in Heston model with 5 jumps and its increments (right).}
\end{figure}

Deciding whether there are jumps in a price process or if it is continuous is beyond volatility estimation one of the most important problems in the literature on high-frequency data. Like ``\emph{Cat or dog?}'' nowadays seems to be one of the big questions 
based on inputs from images in research on AI and machine learning, ``\emph{Jump or no jump?}'' poses one of the main testing problems for statistics of financial markets. One reason is that it is crucial to select and work with an adequate price model. Moreover, volatility of continuous price movements and jumps are used to describe different kinds of market risks and it is important to distinguish between the two in economic studies. Based on \emph{continuous-time} paths it was simply possible to see jumps, while based on \emph{discrete} recordings with a fix time $\Delta$ between subsequent observations it is impossible to decide about the question. In a high-frequency asymptotic regime with distance $\Delta_n\to 0$ between discrete recordings, the question yields an interesting problem. In Section \ref{sec:3} we address this problem based on methods from extreme value theory and the beautiful theory of order statistics placing R\'{e}nyi's representation at the forefront. Extreme value theory concerns outliers in stationary sequences of random variables or time series, in particular the asymptotic distribution of maxima of i.i.d.\ random variables. An outlier refers to a realized value which is far away from the mean level of a time series, e.g., a year with a once-in-a-century flood in a sequence of yearly rainfall data. The mean level before and after an outlier remains the same. The picture of a jump is different. For instance, if a stock price evolves around level $a$ before a jump of size $b$ in response to the communication of quarterly results of a company, the price will continue to move at the new level $a+b$ after the event. Jumps and outliers are nevertheless closely related. In particular, as Figure \ref{Fig:jumps} reveals, jumps in a process can be detected as outliers in the sequence of corresponding increments. Applying methods from extreme value theory to increments is hence one promising starting point to construct jump tests. Figure \ref{Fig:jumps} shows a simulated price based on a popular Heston model with constant drift and its specific stochastic volatility $(\sigma_t)$, to that we add compound Poisson jumps with jump sizes drawn from a Laplace distribution. Even though the five realized jumps are highlighted in the path of the process left-hand side in  Figure \ref{Fig:jumps}, it might be difficult to spot all of them, while it is easier to find the increments with jumps as outliers in the increments right-hand side.

In view of numerous references on the broader topic\,\footnote{Searching for the keyword ``high frequency data'' yields ca.\ 8 million results.}, we only strive to sketch a picture of three selected recent developments. Concerning deeper, open questions, parts of the more recent literature in the field has become rather technical and sophisticated. It is the goal here to shed light on some key ideas and insights, and to point out their relation to concepts from probability and mathematical statistics. Our journey shall {\emph{not}} be a random walk. We begin in Section \ref{sec:2} as a prologue with fundamental concepts and cornerstones of statistics for semimartingales based on high-frequency data. We highlight the impact of \emph{Jean Jacod}. Section \ref{sec:3} reviews statistical tests for price jumps in high-frequency data based on extreme value theory. We complement the existing methods proposing alternative, original methods. \emph{Gumbel} is the second researcher whose contribution is emphasized. The newly developed R\'{e}nyi test in Theorem 1 is built upon a maximal difference between order statistics motivated by the \emph{R\'{e}nyi representation}. This is the first classical, elementary yet ingenious, concept from probability which we highlight in a colourbox. A comment on the genius behind this is given in another box. In Section \ref{sec:4} we provide a brief review on rough fractional stochastic volatility with a data example. The review mentions recent results on the identifiability of the Hurst exponent under high-frequency asymptotics. Our Theorem 2 adds a new result which points out that the regularity of stochastic volatility in a more general sense is identifiable only in some cases and can be estimated only with a slower rate of convergence. This reveals that -- even based on high-frequency data -- there are frontiers in recovering path properties of a latent volatility from price recordings. The third part of the journey in Section \ref{sec:5} is admittedly captured to summarize and reflect some own research. The colourboxes place the spotlight on the taxi problem and the reflection principle for Brownian motion. The taxi problem is a popular example for estimating a boundary parameter. It is even contained in some school-books which develop the main estimation ideas in an intuitive manner. A nice feature of this example is that it is nevertheless deeper understood with concepts as complete and sufficient statistics. While the taxi problem is used to motivate the estimation based on local minima of best ask quotes, the reflection principle paves the way to determine distributions of functionals of Brownian motion which is an important ingredient of the asymptotic analysis of our boundary model.

 
\section{Elements of high-frequency statistics for semimartingales\label{sec:2}}
For some first simple but useful insights consider the parametric model with log-price
\[X_t=\mu\cdot t+\sigma\cdot W_t\,,\]
with a standard Brownian motion $(W_t)$ and $\mu\in\R$, $\sigma>0$ unknown parameters. An obvious problem for statistics is to estimate the two parameters. Assume that we have discrete observations $X_0,X_{t_1},\ldots,X_{t_n}$ on an equidistant grid $t_j=j\Delta_n$, $0\le j\le n$, available. In statistics for stochastic processes there are different asymptotic regimes, either $T=n\Delta_n$ is fix and $\Delta_n\to 0$ (\emph{high-frequency}), or $\Delta_n=\Delta$ is fix and $T=n\Delta\to\infty$ (\emph{low-frequency}), or even both is true and we have high-frequency data over an asymptotically large time interval. 
We have in our setting many observations of only \emph{one single path} of the process. The inference and testing problems, e.g., whether there are jumps or not, are in this field formulated and addressed \emph{path-wise}, i.e., we want to know if the realized path has jumps or not.

Statistical inference is based on the increments
\begin{align}\label{incr}\Delta_j^n X=X_{t_j}-X_{t_{j-1}},1\le j\le n\,.\end{align}
In the parametric, equidistant setting it holds that $\Delta_j^n X=\mu\cdot \Delta_n+\sqrt{\Delta_n}\cdot\sigma\cdot Z_i$, with $Z_i$ i.i.d.\ standard normal, denoted $\mathcal{N}(0,1)$, random variables, such that the increments have expectation $\mu \Delta_n$ and variance $\sigma^2 \Delta_n$. In this standard model, the maximum likelihood estimator
\begin{subequations}
\begin{align}\hat\mu_{ML}&=\frac{\sum_{j=1}^n\Delta_j^n X}{n\Delta_n}=\frac{X_T-X_0}{T}\,,\\
\hat\sigma^2_{ML}&=\frac{\sum_{j=1}^n(\Delta_j^n X)^2}{n\Delta_n}-(\hat\mu_{ML})^2\Delta_n=\frac{\sum_{j=1}^n(\Delta_j^n X)^2}{n\Delta_n}-\frac{\big(\sum_{j=1}^n\Delta_j^n X\big)^2}{nT}\,,
\end{align}
\end{subequations}
has the smallest possible variance. The quadratic risk of the estimated drift parameter
\begin{align}\label{mse}\E\big[\big(\hat\mu_{ML}-\mu\big)^2\big]=\E\Big[\Big(\frac{\sigma\cdot W_T}{T}\Big)^2\Big]=\frac{\sigma^2}{T}\end{align}
does not depend on $\Delta_n$. It tends to zero only if $T\to\infty$, and not under high-frequency asymptotics.
Modelling intra-daily financial data, we have usually discrete observations available at very high frequencies, e.g.\ once per second. Naturally, we work in a high-frequency asymptotic regime. We learn from \eqref{mse} that we cannot consistently estimate the drift in this situation. Looking at $\hat\sigma^2_{ML}$, we see that the term with $\hat\mu_{ML}$ tends to zero. The standard estimator for $\sigma^2$ is therefore the \emph{realized volatility}
\begin{align}\label{rv}\hat\sigma^2=T^{-1}\sum_{j=1}^n(\Delta_j^n X)^2\,.\end{align}
It is an elementary exercise using moments of a normal distribution to compute its squared risk
\begin{align}\label{mse2}\E\big[\big(\hat\sigma^2-\sigma^2\big)^2\big]=\frac{2\sigma^4\Delta_n}{T}+\frac{4\sigma^2\mu^2\Delta_n^2}{T}+\mu^4\Delta_n^2\,,\end{align}
which tends to zero under high-frequency asymptotics. We should not use it in a low-frequency regime in that its variance tends to zero, but the bias does not. Throughout the remainder of this manuscript we set $T=1$, without loss of generality. It is not surprising that a central limit theorem (clt)
\begin{align}\label{clt1}\sqrt{n}(\hat\sigma^2-\sigma^2)\todl \mathcal{N}(0,2\sigma^4)\,,\end{align}
holds true. We write $\todl $ for convergence in distribution (weak convergence). Note, however, that we have a \emph{triangular array} of random variables here and not a sequence, that is, going from $n$ to $n+1$ is not just adding one observation, but all observation times depend on $n$. For this reason, clts for triangular arrays need to be used in the high-frequency framework. The main benefit of a clt is to facilitate \emph{asymptotic confidence} statements. These are feasible when we standardize the left-hand side in \eqref{clt1} with a consistent estimator of the asymptotic standard deviation. A more elegant method -- which might give a slightly better approximation for finite samples -- is a \emph{variance stabilization} applying the $\Delta$-method with the strictly increasing logarithm:
\begin{align}\label{fclt1}\sqrt{n}(\log\big(\hat\sigma^2\big)-\log(\sigma^2))\todl \mathcal{N}(0,2)\,.\end{align}
Since $\log$ is strictly increasing, confidence intervals for $\log(\sigma^2)$ readily translate into confidence intervals for $\sigma^2$. A current discussion comparing four approaches to perform asymptotic confidence based on a clt, which all work in this example, is given in \cite{politis}. Beyond parameter estimation, the analogy to the standard statistical model allows to transfer more methods, e.g., likelihood ratio tests.

In the more general model with high-frequency observations of a continuous semimartingale $(C_t)$ from \eqref{SM} with time-varying drift $(\mu_s)$ and volatility $(\sigma_s)$, the first goal is estimation of the \emph{integrated volatility} $\int_0^1\sigma_s^2\,\text{d}s$, e.g., integrated over trading days as a daily measure of risk. Since It\^{o}'s isometry yields that
\begin{align*}\E\big[(\Delta_j^n C)^2\big]=\int_{(j-1)\Delta_n}^{j\Delta_n}\sigma_s^2\,\text{d}s+\mathcal{O}(\Delta_n^2)\,,\end{align*}
the realized volatility \eqref{rv} is still a suitable (and in fact optimal) estimator.

A very strong asymptotic result under mild regularity conditions is the \emph{functional stable central limit theorem} for realized volatility by \cite{jacod1}:
\begin{align}\label{fclt2}\sqrt{n}\,\Big(\sum_{j=1}^{\lfloor nt\rfloor }\big(\Delta_j^n C\big)^2-\int_0^t\sigma_s^2\,\text{d}s\Big)\stackrel{st}{\longrightarrow} \int_0^t\sqrt{2}\sigma_s^2\, \text{d}B_s~,t\in[0,1]\,,\end{align}
with $(B_s)$ a Brownian motion independent of $(W_s)$ defined on an extension of the original probability space. This implies the marginal clt
\begin{align}\label{clt2}\sqrt{n}\,\Big(\sum_{j=1}^{n}\big(\Delta_j^n C\big)^2-\int_0^1\sigma_s^2\,\text{d}s\Big)\stackrel{st}{\longrightarrow} \mathcal{N}\Big(0,\int_0^1 2\,\sigma_s^4\, \text{d}s\Big)\,.\end{align}
For stochastic volatility, the variance of the limit distribution is random, the limit is then called \emph{mixed normal}. For this reason it is important that the convergence is \emph{stable}. This is a stronger mode of weak convergence equivalent to joint weak convergence with every measurable bounded random variable on the same space. Since it allows for a $\Delta$-method and weak convergence after standardization, known as Slutsky's lemma for weak convergence, it is a crucial ingredient to construct asymptotic confidence intervals. Beyond inference on the integrated volatility, the functional clt allows for various other statistical applications, for instance, a volatility change-point test of cusum-type as explained in Section 2 of \cite{BJV}.
 
\begin{figure}[t]
{\colorbox{lightgray}{\parbox{0.98\linewidth}{\textbf{Jean Jacod} is certainly a spiritus rector of high-frequency statistics for semimartingales. He was head of the probability group at Paris VI (Pierre et Marie Curie) from 1987 until 2000. He was well known for his textbooks and research on limit theorems, jump processes and Malliavin calculus when he established the main groundwork in this field. He pointed out that stable convergence is the right concept for asymptotic statements on realized volatility and related functionals. Proofs of stable clts for functionals of semimartingale increments usually rely on his results. He moreover provided techniques to separate jumps from continuous movements which are exploited by many authors. The textbooks \cite{JP} and \cite{sahaliajacod} summarize main aspects of high-frequency statistics. 
}}}
\end{figure}

It is not only relevant to infer the integrated volatility, the nonparametric estimation of the \emph{spot volatility} process $(\sigma_s^2)$ is another central problem in high-frequency statistics. At this point, it is beneficial to introduce some rigorous assumption on the characteristics of the continuous semimartingale log-price process $(C_t)$.
\begin{assump}\label{volass}
The drift $(\mu_t)_{t\ge 0}$ is locally bounded and the volatility is strictly positive, $\inf_{t\in[0,1]}\sigma_t>0$, almost surely. 
For all $0\leq t+s\leq1$, $t\ge 0$, $s\ge 0$, with some constants $ C_{\sigma}>0$, and $\alpha>0$, it holds that 
\begin{align}\label{vola}\E\big[(\sigma_{(t+s)}-\sigma_{t})^2\big]  \le C_{\sigma} s^{2\alpha}\,.\end{align}
\end{assump}
Condition \eqref{vola} imposes a certain \emph{regularity} $\alpha$ of the volatility process. Due to the expectation, it is not Hölder continuity and \eqref{vola} does not rule out volatility jumps. The increments of some compound Poisson jump process for instance, over a time interval of length $s$, equal a constant times $(s^2+s)$, if the jump size distribution has a second moment. Therefore, it satisfies \eqref{vola} with $\alpha=1/2$. This is true for much more general jump processes. A continuous semimartingale and in particular Brownian motion satisfy \eqref{vola} with $\alpha=1/2$. The Hölder condition \eqref{vola} in quadratic mean is a convenient concept to describe the variability of a stochastic process. It is also used in other fields of probability, for instance, for functional data in \cite{localreg}. In particular, the rate of convergence, with that $(\sigma_s^2)$ at some time $s\in(0,1)$ can be estimated, hinges on the regularity parameter $\alpha$. Using a local average of rescaled squared increments
\begin{align}\label{spotvola}\hat\sigma_{s}^2=\frac{n}{k_n} \sum_{j=\lfloor s n\rfloor +1}^{\lfloor s n\rfloor +k_n} \big(\Delta_j^n C\big)^2\,,\end{align}
as estimator yields with $k_n=n^{2\alpha/(2\alpha+1)}$, for which the order of the squared bias $(k_n\Delta_n)^{2\alpha}$ is the same as that of the variance $k_n^{-1}$, the minimal root mean squared error of order $n^{-\alpha/(2\alpha+1)}$. Given that the regularity parameter determines optimal spot volatility estimation, inference on an unknown $\alpha$ is certainly of interest. This, however, is an intricate problem not yet solved in general which we visit in Section \ref{sec:4}. In the standard case $\alpha=1/2$, spot volatility can be estimated with rate $n^{-1/4}$. In the best (non-constant) case $\alpha=1$, the rate is $n^{-1/3}$.

If there are jumps, the realized volatility converges in probability, denoted $\stackrel{\P}{\longrightarrow}$, to the entire quadratic variation:
\[\sum_{j=1}^{n}\big(\Delta_j^n X\big)^2\stackrel{\P}{\longrightarrow}\int_0^1\sigma_s^2\text{d}s+\sum_{s\le 1}(\Delta X_s)^2\,.\]
We write $\Delta X_s=X_s-X_{s-}=X_s-\lim_{t\uparrow s} X_t$. It is common notation to write jumps with sums over uncountable index sets, since the processes will always have only countably many random jump times. Most relevant is first to estimate the integrated volatility in presence of the \emph{nuisance} jumps. To get rid of jumps, we need to discard in particular large jumps as in Figure \ref{Fig:jumps} which are contained in the large absolute increments. A natural approach is hence to truncate increments whose absolute values are above a certain threshold. For this purpose, define a sequence $(u_n)$, with $u_n\propto \Delta_n^{\tau}$, $\tau\in(0,1/2)$. Denote $\1\{A\}$ an indicator function which is 1 if $A$ is true and 0, else. The idea is now to work out under which restrictions on the jumps the \emph{truncated realized volatility} satisfies 
\begin{align}\label{clt3}\sqrt{n}\Big(\sum_{j=1}^{n}\big(\Delta_j^n X\big)^2\1\{|\Delta_j^n X|\le u_n\}-\int_0^1\sigma_s^2\,\text{d}s\Big)\stackrel{st}{\longrightarrow} \mathcal{N}\Big(0,\int_0^1 2\,\sigma_s^4\, \text{d}s\Big)\,,\end{align}
the same clt as for the realized volatility without jumps in \eqref{clt2}. Sufficient is
\begin{align*}\sqrt{n}\Big(\sum_{j=1}^{n}\big(\Delta_j^n X\big)^2\1\{|\Delta_j^n X|\le u_n\}-\sum_{j=1}^{n}\big(\Delta_j^n C\big)^2\Big)\stackrel{\P}{\longrightarrow} 0\,.\end{align*}
The truncation method was pioneered by \cite{mancini}, is summarized in Chapter 13 of \cite{JP}, and the community is still working on refinements, see, for instance, \cite{figueroa2019optimum} and \cite{amorino2020unbiased}.
We can decompose the difference 
\begin{align*}&\sum_{j=1}^{n}\big(\Delta_j^n X\big)^2\1\{|\Delta_j^n X|\le u_n\}-\sum_{j=1}^{n}\big(\Delta_j^n C\big)^2\\
&=\sum_{j=1}^{n}\1\{|\Delta_j^n C|>\kappa u_n\}\Big(\big(\Delta_j^n X\big)^2\1\{|\Delta_j^n X|\le u_n\}-\big(\Delta_j^n C\big)^2\Big)\\
&\quad - \sum_{j=1}^{n}\1\{|\Delta_j^n C|\le \kappa u_n,|\Delta_j^n X|> u_n\}\big(\Delta_j^n C\big)^2\\
&\quad + \sum_{j=1}^{n}\1\{|\Delta_j^n C|\le \kappa u_n,|\Delta_j^n X|\le u_n\}\big(\big(\Delta_j^n X\big)^2-(\Delta_j^n C\big)^2\big)\\
&=\text{\RomanNumeralCaps{1}}_n+\text{\RomanNumeralCaps{2}}_n+\text{\RomanNumeralCaps{3}}_n\,,
\end{align*}
with some $\kappa\in(0,1)$, e.g., $\kappa=1/2$. For the term $\text{\RomanNumeralCaps{1}}_n$, knowing that $\max_j|\Delta_j^n C|$ is of stochastic order $\log(n)\sqrt{\Delta_n}$, what is contained in the next section, and that it is even almost surely smaller than $u_n$ by the law of the iterated logarithm is sufficient. The argument by Jacod is more elementary and in the following way:
\begin{align*}|\text{\RomanNumeralCaps{1}}_n|&\le \sum_{j=1}^{n}\1\{|\Delta_j^n C|>\kappa u_n\}\Big(\big(\Delta_j^n X\big)^2\1\{|\Delta_j^n X|\le u_n\}+\big(\Delta_j^n C\big)^2\Big)\\
&\le (1+\kappa^{-2})\sum_{j=1}^{n}\big(\Delta_j^n C\big)^2\bigg|\frac{\Delta_j^n C}{\kappa u_n}\bigg|^N\,,
\end{align*}
for any $N\in\N$, and using moments of $(\Delta_j^n C)$ and choosing $N$ sufficiently large yields that $\sqrt{n}\,\text{\RomanNumeralCaps{1}}_n\stackrel{\P}{\longrightarrow} 0$, since $\sqrt{n}\,\E\big[|\text{\RomanNumeralCaps{1}}_n|\big]\to 0$. 

The two other terms require a closer look at the jump component $(J_t)$ of the semimartingale $(X_t)$, $X_t=C_t+J_t$. Most literature in the area imposes the general structure of It\^{o} semimartingales in the sense of Section 2.1.4 of \cite{JP} which admit a ``Grigelionis representation''. The jumps are independent of $(C_t)$ and separated in a martingale of compensated small jumps and larger jumps
\begin{align*}J_t\negthinspace= \negthinspace\int_0^t\negthinspace \int_{\mathbb{R}}\delta(s,z)\1{\{|\delta(s,z)|\leq 1 \}}(\mu-\nu)(\text{d}s,\text{d}z) \negthinspace+\negthinspace\int_0^t \negthinspace\int_{\mathbb{R}}\delta(s,z)\1{\{|\delta(s,z)|> 1 \}} \mu(\text{d}s,\text{d}z), \end{align*}
with a Poisson random measure $\mu$ compensated by $\nu(\text{d}s,\text{d}z)=\lambda(\text{d}z)\otimes \text{d}s$, with a $\sigma$-finite measure $\lambda$. The function $\delta$, for which the third argument $\omega$ is consequently also not written, is a predictable function, for which we assume that a non-negative deterministic function $\gamma$ exists, such that
\begin{align*}\sup_{\omega,x}|\delta(t,x)|/\gamma(x)\end{align*} 
is locally bounded. The benefit of such a meticulous definition of $(J_t)$ is to preserve generality. For instance, the large class of L\'{e}vy jump processes is contained as a special case with $\delta(s,x)=x$. The main assumption on the jumps is captured in a {\emph{jump activity index}} $r\in[0,2]$, for which
\begin{align}\label{BG}\int_{\mathbb{R}}(\gamma^{r}(z)\wedge 1)\lambda(\text{d}z)<\infty\,.\end{align}
Basically, this means summability of $\sum_{s\le 1}|\Delta J_s|^r$. For $r=0$ this is a strong restriction with at most finitely many jumps on $[0,1]$ (\emph{finite-activity}), and for $r=1$ we assume the jump process to be of \emph{finite variation}. The larger $r$, the less restrictive is the condition.

Now we have the toolbox to handle terms $\text{\RomanNumeralCaps{2}}_n$ and $\text{\RomanNumeralCaps{3}}_n$. In the sequel, let $K$ be a generic constant. With Markov's inequality and Cauchy-Schwarz, we obtain for $\text{\RomanNumeralCaps{2}}_n$  that
\begin{align*}
&\E[|\text{\RomanNumeralCaps{2}}_n|]\le K \sum_{j=1}^{n}\Delta_n \P\big(|\Delta_j^n J|\ge (1-\kappa) u_n\big)\\
&\le K \sum_{j=1}^{n}\Delta_n u_n^{-2}(1-\kappa)^{-2}\E\big[|\Delta_j^n J|^2\1\{|\Delta_j^n J|\ge  (1-\kappa) u_n\}\big]\\
&\le K \sum_{j=1}^{n}\Big(\int_{(j-1)\Delta_n}^{j\Delta_n}\int_{\R}\gamma(z)\1\{\gamma(z)>u_n\}\nu(\text{d}s,\text{d}z)\Big)^2\\
&\le K \sum_{j=1}^{n}\Big(\int_{(j-1)\Delta_n}^{j\Delta_n}\int_{\R}\gamma^2(z)\1\{\gamma(z)>u_n\}\text{d}s\lambda(\text{d}z)\int_{(j-1)\Delta_n}^{j\Delta_n}\int_{\R}u_n^{-r}\gamma^r(z)\text{d}s\lambda(\text{d}z)\Big)\\
&=\mathcal{O}\big(\Delta_n^{1-r\tau}\big)\,.
\end{align*}
The term $\text{\RomanNumeralCaps{2}}_n$ can be thought of as an error by truncating also the continuous components whenever the threshold is exceeded. Therefore, the order gets smaller for smaller $\tau$, moving $u_n$ farer away from $\sqrt{\Delta_n}$, when only very large absolute increments are truncated. If $r\tau<1/2$, we conclude that $\sqrt{n}\text{\RomanNumeralCaps{2}}_n\stackrel{\P}{\longrightarrow} 0$.

Most difficult is term $\text{\RomanNumeralCaps{3}}_n$ due to small jumps in non-truncated increments. We exploit the martingale structure of the small jumps to apply Burkholder's inequality and to deduce
\begin{align*}
&\E[|\text{\RomanNumeralCaps{3}}_n|]\le K \sum_{j=1}^{n}\E\big[|\Delta_j^n J|^2\1\{|\Delta_j^n J|\le  (1+\kappa) u_n\}\big]\\
&\le K \sum_{j=1}^{n}\E\big[([J,J]_{j\Delta_n}-[J,J]_{(j-1)\Delta_n})\1\{|\Delta_j^n J|\le  (1+\kappa) u_n\}\big]\\
&\le K \sum_{j=1}^{n}\int_{(j-1)\Delta_n}^{j\Delta_n}\int_{\R}\gamma^2{(z)}\1\{\gamma(z)\le u_n\}\nu(\text{d}s,\text{d}z)\\
&\le K \sum_{j=1}^{n}u_n^{2-r}\Delta_n \int_{\R}\gamma^r{(z)}\lambda(\text{d}z)\\
&=\mathcal{O}\big(\Delta_n^{\tau(2-r)}\big)\,.
\end{align*}
This term gets smaller for larger $\tau$, moving $u_n$ closer to $\sqrt{\Delta_n}$. Both terms decrease for smaller $r$. 
To ensure that $\sqrt{n}\text{\RomanNumeralCaps{3}}_n\stackrel{\P}{\longrightarrow} 0$, we need that $r<1$, and $\tau(2-r)>1/2$. This is the main result about truncated realized volatility: If $r<1$, with $\tau\in\big((2(2-r))^{-1},1/2\big)$, what can be ensured by selecting $\tau$ close to $1/2$, it satisfies \eqref{clt3}. While we emphasize some key steps of the proof, the bounds for $\text{\RomanNumeralCaps{2}}_n$ and $\text{\RomanNumeralCaps{3}}_n$ admittedly lack some details. Most of them are elementary, as carefully using the triangle inequality, but a few are deeper. A less pedagogic but rigorous proof can be found in Chapter 13 of \cite{JP}. In particular, in the bound for $\text{\RomanNumeralCaps{2}}_n$ we work under the event with at most one larger jump contained in one increment, for which the Poisson nature of the jumps yields precise estimates, see Step 5 in the proof of Thm.\ 13.1.1 in \cite{JP}. The restrictions on $r$ for spot volatility estimation with truncation are less strict, $r<4/3$, see Section 13.4.1 of \cite{JP}, mainly since the rate is slower with that such a difference needs to tend to 0.

\section{Jump detection in high-frequency data based on extreme value theory\label{sec:3}}

There are several different constructions of tests for jumps in high-frequency data. Let me focus here only on the most prominent one in financial economics by \cite{leemykland}. It is based on the maximal (absolute) normalized increment and exploits its asymptotic Gumbel distribution under the null hypothesis of no jumps. The test is sometimes called the \emph{Lee-Mykland test}, or the \emph{Gumbel test} in the literature. The asymptotic Gumbel distribution is traced back to the one of the maximum of i.i.d.\ $\mathcal{N}(0,1)$ random variables and thus classical extreme value theory.

In the sequel, we consider real-valued random variables on some probability space with measure $\P$. For random variables $(X_j)_{1\le j\le n}$, we denote the \emph{order statistics} $X_{(j)},1\le j\le n$, with $X_{(1)}\le X_{(2)}\le \ldots,X_{(n)}$. In particular, $X_{(1)}=\min_{1\le j\le n} X_j$ refers to the minimum and $X_{(n)}=\max_{1\le j\le n} X_j$ to the maximum. These are unique with probability 1 for random variables whose distributions are absolutely continuous with respect to the Lebesgue measure. It is a standard example in extreme value theory, see, e.g., Example 1.1.7 in \cite{haan}, that for $(X_j)_{1\le j\le n}$ i.i.d.\ $\mathcal{N}(0,1)$, the maximum satisfies
\begin{align}\label{gumbel}
a_n^{-1}\big(X_{(n)}-b_n\big)\todl \Lambda\,,a_n^{-1}=\sqrt{2\log (n)}~, ~\mbox{and}~b_n=a_n^{-1}-\frac{\log(4\pi\log (n))}{2\sqrt{2\log (n)}}\,,
\end{align}
with the \emph{Gumbel limit distribution} $\Lambda$, i.e., it holds for all $x\in\R$ that 
\begin{align*}
\lim_{n\to\infty}\P\Big(a_n^{-1}\Big(X_{(n)}-b_n\Big)\Big)\le x\Big)= \exp\big({-e^{-x}}\big)\,.
\end{align*}
Limit distributions of maxima of i.i.d.\ random variables can only be of Gumbel, Weibull or Frech\'{e}t type. If sequences $(a_n)$ and $(b_n)$ exist, such that for a cumulative distribution function (cdf) $F$ the maxima of i.i.d.\ random variables with cdf $F$ converge to $\Lambda$, write $F\in \text{MDA}(\Lambda)$ (maximum domain of attraction).

\begin{figure}[t]
{\colorbox{lightgray}{\parbox{0.98\linewidth}{\textbf{Emil Julius Gumbel} is an interesting personality for at least two reasons, his work as a researcher on extreme value theory, and his political commitment. Despite his academic contributions, his pacifism, statistical research which exposed the leniency towards political murders committed by right-wing extremists, and his political activity in general led to his dismissal from the University of Heidelberg in 1932. This troubling chapter in the university's history is recently critically reflected.\,\footnotemark[3] Both aspects of his legacy arouse interest in the recent literature on the history of science, see, e.g, \cite{vogt2021} and \cite{rendtel2021} (in German).}}}
\end{figure}
\footnotetext[3]{\href{https://www.uni-heidelberg.de/de/universitaet/heidelberger-profile/historische-portraets/krieg-gegen-einen-pazifisten}{Emil Julius Gumbel: Krieg gegen einen Pazifisten} (in German).}
\addtocounter{footnote}{1}

On the null hypothesis $H_0:\sup_{\tau\in (0,1)}|X_{\tau}- X_{\tau-}|=0$, based on high-frequency returns $(\Delta_j^n X)_{1\le j\le n}$, \cite{leemykland} proved in their Thm.\ 2 that 
\begin{align}\label{gumbellm}\sqrt{2\log (2n)}\Big(n^{1/2}\,\max_{1\le j\le n}\bigg|\frac{\Delta_j^n X}{\hat\sigma_{j/n}}\bigg|-\Big(\sqrt{2\log (2n)}-\frac{\log(4\pi\log (2n))}{2\sqrt{2\log (2n)}}\Big)\Big)\todl\Lambda,\end{align}
with a suitable estimator of the volatility $(\hat\sigma_{j/n})$, e.g., from \eqref{spotvola}. The proof is carried out under some assumptions on $(\mu_t)$ and $(\sigma_t)$, which can be generalized to rather weak regularity conditions. The similarity to \eqref{gumbel} is striking. Indeed, the proof traces back the convergence to \eqref{gumbel} showing that the normalized increments can be approximated by i.i.d.\ $\mathcal{N}(0,1)$ observations. The factor $2$ in the logarithm is due to the absolute value in the statistic and exploits the symmetry of $\mathcal{N}(0,1)$. The normalizing sequences given in \cite{leemykland} are in fact slightly different, but asymptotically equivalent. Rejecting the null hypothesis when the statistic left-hand side in \eqref{gumbellm} exceeds $-\log(-\log(1-\alpha))$, $\alpha\in(0,1)$, hence yields a test with \emph{asymptotic level} $\alpha$, i.e., the probability of a false rejection converges to $\alpha$. Under the alternative hypothesis \(H_1:\sup_{\tau\in(0,1)}|X_{\tau}- X_{\tau-}|>0\), the test rejects correctly with asymptotic probability 1. There is moreover a rate of convergence. We can state equivalently that the test rejects correctly with asymptotic probability 1 under \emph{local alternatives}
\[H_1:\liminf_{n\to\infty}{\dbl{n^{\beta}}}\sup_{\tau\in(0,1)}|X_{\tau}- X_{\tau-}|>0,~\mbox{for some}~{\dbl{\beta<1/2}}\,.\]
This means that we can not detect arbitrarily small jumps based on a fix number of $(n+1)$ discrete high-frequency recordings, but jumps which are larger than of order $n^{-1/2}$. The test has several appealing properties. Critical values based on quantiles of the Gumbel distribution can be determined to test at a chosen level $\alpha$. Moreover, the associated argmaximum consistently estimates the time of the largest jump under $H_1$. Based on sequential testing and the largest absolute increments thus jump times and jump sizes can be inferred.

The next paragraph advances research on high-frequency jump tests contributing alternative methods based on extreme value theory which have some advantages compared to the Gumbel test. For the construction, we make an excursion to the nice, classical theory of order statistics. The \emph{exponential distribution}, $\text{Exp}(\lambda)$, with parameter $\lambda>0$, with Lebesgue density $\lambda\exp(-\lambda t),t>0$, and tail function $\P(X>t)=\exp(-\lambda t),t>0$, $X\sim \text{Exp}(\lambda)$, takes an outstanding role in the theory of order statistics.

\begin{figure}[t]
{\colorbox{light-blue}{\parbox{0.98\linewidth}{
\textbf{R\'{e}nyi's representation} is a key result about order statistics.
\begin{lem}\label{renyi}
Let $(E_j)_{1\le j\le n}$ be i.i.d.\ $\text{Exp}(1)$. The equality in distribution
\begin{align*}\big(E_{(1)},E_{(2)},\ldots,E_{(n)}\big)\stackrel{d}{=}\Big(\frac{\tilde E_1}{n},\frac{\tilde E_1}{n}+\frac{\tilde E_2}{n-1},\ldots,\frac{\tilde E_1}{n}+\frac{\tilde E_2}{n-1}+\ldots +\frac{\tilde E_n}{1}\Big)\end{align*}
holds for $(\tilde E_j)$ i.i.d.\ $\text{Exp}(1)$.
\end{lem}
This shows that differences of subsequent order statistics of i.i.d.\ $\text{Exp}(1)$ random variables are independent and
\[E_{(k)}-E_{(k-1)}\stackrel{d}{=}\frac{\tilde E_k}{n-k+1}\sim \text{Exp}(n-k+1),\;k=1,\ldots,n,E_{(0)}=0\,.\]
While the proof is typically based on an elementary change of variables, the result is deeply rooted in the characteristic memorylessness property of the exponential distribution: Conditional on an event $\{E_1>t\},\, t>0$, the tail probability
\[\P(E_1>t+s|E_1>t)=\frac{\P(E_1>t+s)}{\P(E_1>t)}=\frac{\exp(-(t+s))}{\exp(-t)}=\exp(-s)\,,s>0,\]
shows that the conditional distribution is again $\text{Exp}(1)$. The exponential distribution is moreover min-stable, i.e., $E_{(1)}\sim \text{Exp}(n)$, since 
\[\P(E_{(1)}>t)=\P(\cap_{ i=1}^{n}E_i >t)= (\P(E_1 >t))^n=\exp(-n\cdot t)\,,t>0\,.\]
Due to the memorylessness the difference $(E_{(2)}-E_{(1)})$ is $\text{Exp}(n-1)$-distributed as the minimum of $(n-1)$ independent $\text{Exp}(1)$ random variables and Lemma \ref{renyi} follows by induction. Working with order statistics of general i.i.d.\ random variables, we typically apply transformations to the exponential distribution to exploit R\'{e}nyi's representation. 
}}}
\end{figure}
\begin{figure}[t]
{\colorbox{lightgray}{\parbox{0.98\linewidth}{I expect that \textbf{Alfr\'{e}d R\'{e}nyi} is well known to many readers. A list of his various contributions to number theory, probability, analysis and to many more mathematical fields in memoriam of him is given in \cite{renyiref}. Let me highlight the eminent importance of R\'{e}nyi's representation from \cite{renyi}. It is often exploited, e.g., for the asymptotic analysis of estimators of the extreme value index, various statistical methods and for the presented test in the summarized recent area of research.}}}
\end{figure}

We build up our methods on the joint asymptotic distribution of the extreme order statistics. 
\begin{prop}\label{order}Let $(X_j)_{1\le j\le n}$ be i.i.d.\ real-valued random variables with $X_j\sim F\in \text{MDA}(\Lambda)$. For fix $r$, there exist sequences $(a_n)$ and $(b_n)$, such that as $n\rightarrow\infty$, it holds that
\begin{align*}\left(\begin{array}{c}a_n^{-1}(X_{(n)}-b_n)\\ a_n^{-1}(X_{(n-1)}-b_n)\\ \vdots \\ a_n^{-1}(X_{(n-r+1)}-b_n)\end{array}\right)\todl \left(\begin{array}{c}-\log({E_1})\\ -\log( E_1+ E_2)\\ \vdots \\ -\log(E_1+\ldots+ E_{r})\end{array}\right)\,,\end{align*}
where $(E_j)$ are i.i.d.\ $\text{Exp}(1)$.
\end{prop}
This is a special case of Thm.\ 2.1.1 from \cite{haan} for distributions in the MDA of a Gumbel distribution. If $X_j\sim\mathcal{N}(0,1)$, the sequences $(a_n)$ and $(b_n)$ coincide with the ones from \eqref{gumbel}. In particular, $-\log({E_1})$ has a Gumbel distribution. The main ingredient of the proof of Proposition \ref{order} is the convergence in distribution
\[n\,\big(\tilde E_{(1)},\tilde E_{(2)},\ldots,\tilde E_{(r)}\big)\stackrel{d}{\longrightarrow}\big( E_1, E_1+E_2,\ldots, \sum_{j=1}^r E_j\big)\,,\,n\to\infty, r\;\text{fix},\]
for $(\tilde E_j)$ i.i.d.\ $\text{Exp}(1)$, which is directly implied by R\'{e}nyi's representation. With a standard analytical condition for extreme value convergence, a change of variables and an (extended) continuous mapping theorem this yields the result.

Based on Proposition \ref{order}, we derive that
\begin{align}\label{exptest}\sqrt{2\log(n)}\big(X_{(n)}-X_{(n-r)}\big)\todl \log\Big(\frac{E_1+\ldots E_{r+1}}{E_1}\Big)\,.\end{align}
This motivates an interesting alternative to the Gumbel test to construct a test based on \emph{differences} of ordered normalized increments related to the distribution of $(X_{(n)}-X_{(n-r)})$, e.g., for $r=1$. Under jumps from a distribution with a Lebesgue density such a test will attain analogous asymptotic properties. The asymptotic distribution under the null hypothesis is, however, simpler, as it does not require the sequence $(b_n)$ any more. In view of an arduous discussion about the finite-sample fit of different, asymptotically equivalent variants of $(b_n)$, and that incorrect normalizations of the Gumbel test led to some problems in the applied literature, cf.\ \cite{nunes}, the advantage of getting rid of $(b_n)$ in determining critical values of a jump test should not be underestimated.

More reasons to explore this path are in the beauty of the \emph{joint} limit distribution of differences of extreme order statistics, again related to R\'{e}nyi's representation, and an improved, simplified detection of jumps under the alternative hypothesis. The latter implies a practical improvement compared to a sequential application of the Gumbel test which I expect to be of relevance for the current analysis of high-frequency data. We establish the main result along three auxiliary lemmas on the limit distributions in Proposition \ref{order} and \eqref{exptest}. They are suitable as exercises in courses on probability and analysis.

Our first auxiliary result shows that interesting transformations of exponential random variables yield again exponential distributions. 
\begin{lem}Let $E_1,\ldots,E_N$  be i.i.d.\ $\text{Exp}(1)$. Then
\begin{align}\label{Exp}\log\Big(1+\frac{E_{r+1}}{\sum_{j=1}^r E_j}\Big)\sim \text{Exp}(r)\,, \end{align}
for any $r$, $1\le r\le N-1$. In particular, 
\begin{align}\log\Big(1+\frac{E_{2}}{E_1}\Big)\sim \text{Exp}(1)\,. \end{align}
\end{lem}
\begin{proof}
For some non-negative, independent random variables $X$ and $Y$, with Lebesgue densities $f_X$ and $f_Y$, the change of variables
\begin{align*}
\P(X/Y\le z)&=\int_0^{\infty}\int_0^{\infty}\1\{x/y\le z\}f_X(x)f_Y(y)\text{d} x\text{d} y\\
&=\int_0^{\infty}\1\{q\le z\}\Big(\int_0^{\infty}f_X(qy)f_Y(y)y\,\text{d} y\Big)\text{d} q\,,z>0,
\end{align*}
yields the Lebesgue density of the ratio $X/Y$. With the density of the \text{Gamma}($r$,1) distribution of $\sum_{j=1}^r E_j$, i.e., the $r$th convolution of $\text{Exp}(1)$, and independence, we obtain for the density $g$ of $E_{r+1}/\sum_{j=1}^r E_j$:
\begin{align*}
g(z)&=\int_0^{\infty}\exp(-yz)\frac{y^{r-1}}{\Gamma(r)}\exp(-y)y\,\text{d}y\\
&=\int_0^{\infty}\exp(-y(z+1))\frac{y^{r}}{\Gamma(r)}\text{d}y\\
&=\frac{r}{(z+1)^{r+1}}\,,z>0.
\end{align*}
The last identity is implied by the known moments of an exponential distribution with $\lambda=(z+1)$. Since $z\mapsto \log(1+z)$ has inverse $u\mapsto \exp(u)-1$, with derivative $\exp(u)$, a change of variables yields that $U=\log\Big(1+\frac{E_{r+1}}{\sum_{j=1}^r E_j}\Big)$ has the Lebesgue density
\begin{align*}
f_U(u)=\exp(u)\cdot\frac{r}{\exp(u(r+1))}=r\cdot\exp(-ru)\,,u>0.
\end{align*}
Hence, $U\sim \text{Exp}(r)$.
\end{proof}
I find it even more interesting that, although the same random variables enter the transformation \eqref{Exp} for different $r$, we have an independence based on the next two auxiliary lemmas. 
\begin{lem}\label{indstart}For $E_1,E_2,E_3$ i.i.d.\ $\text{Exp}(1)$, the random variables $E_1/(E_1+E_2)$, $(E_1+E_2)/(E_1+E_2+E_3)$, and $(E_1+E_2+E_3)$ are independent.
\end{lem}
\begin{proof}
We show that the joint density equals the product of the marginal densities. Elementary computations yield the Jacobian of the inverse map
\begin{align*}\left(\begin{array}{c}u\\ v\\ w\end{array}\right)\mapsto \left(\begin{array}{c}u\cdot v\cdot w\\ (1-u)\cdot v\cdot w\\ (1-v)\cdot w\end{array}\right)\,,\end{align*}
and its determinant $vw^2$. Based on a (multivariate) change of variables and with the product exponential density of $(E_1,E_2,E_3)$, we obtain the joint density
\begin{align*}
f_{U,V,W}(u,v,w)&=\exp(-uvw)\exp(-(1-u)vw)\exp(-(1-v)w)vw^2\\
&=\exp(-w)vw^2\,,0\le u,v\le 1,w>0\,.
\end{align*}
Since this equals the product of the marginal densities $f_W(w)=w^2e^{-w}/2,w>0$, of the \text{Gamma}(3,1) distribution of $(E_1+E_2+E_3)$, $f_V(v)=2v,v\in[0,1]$, of $(E_1+E_2)/(E_1+E_2+E_3)$, and $f_U(u)=\1\{u\in[0,1]\}$, of the uniform distribution of $E_1/(E_1+E_2)$, we conclude the independence.
\end{proof}
Transformations of independent random variables remain independent. For the general conclusion, we only need to extend Lemma \ref{indstart} what can be done by induction.
\begin{lem}\label{indstep}For $E_1,\ldots,E_{r+1}$, $r\in\N$, i.i.d.\ $\text{Exp}(1)$, the random variables $E_1/(E_1+E_2)$, $\ldots$, $(\sum_{j=1}^r E_j)/(\sum_{j=1}^{r+1}E_j)$, and $(\sum_{j=1}^{r+1}E_j)$ are independent.
\end{lem}
\begin{proof}
From the inverse map
\begin{align*}\left(\begin{array}{c}u_1\\ u_2\\ \vdots \\ u_r\\ u_{r+1}\end{array}\right)\mapsto \left(\begin{array}{c}\prod_{j=1}^{r+1}u_j\\ (1-u_1)\prod_{j=2}^{r+1}u_j\\ \vdots\\ (1-u_{r-1})\cdot u_r\cdot u_{r+1}\\ (1-u_{r})\cdot u_{r+1}\end{array}\right)\,,\end{align*}
we infer by induction the Jacobian
\begin{align*}
J^{r+1,r+1}=\left(\begin{array}{cccc}J_{11}^{r,r}\cdot u_{r+1}&\ldots&J_{1r}^{r,r}\cdot u_{r+1}&\prod_{j=1}^r u_j\\ J_{21}^{r,r}\cdot u_{r+1}&\ldots & J_{2r}^{r,r}\cdot u_{r+1}&(1-u_1)\prod_{j=2}^ru_j\\ \vdots&&\vdots&\vdots\\ J_{r1}^{r,r}\cdot u_{r+1}&\ldots &J_{rr}^{r,r}\cdot u_{r+1}&(1-u_{r-1})\cdot u_r\\ 0&\ldots & -u_{r+1}&(1-u_r)\end{array}\right)\,.
\end{align*}
We write $A_{ij}$ for the entry in the $i$th row and $j$th column of some matrix $A$. Based on a Laplace expansion with respect to the last line, we obtain
\begin{align*}
\det J^{r+1,r+1}&=(1-u_r)\cdot u_{r+1}^r\cdot \det J^{r,r}+u_{r+1}\cdot u_r\cdot u_{r+1}^{r-1}\cdot \det J^{r,r}\\
&=u_{r+1}^r\cdot \det J^{r,r}=\prod_{j=1}^r u_j^{j-1} \cdot u_{r+1}^r\,.
\end{align*}
With a telescoping sum in the exponent, similar as in the proof of Lemma \ref{indstart}, we obtain the joint density $g$ of $U_1,\ldots,U_{r+1}$:
\begin{align*}
g(u_1,\ldots,u_{r+1})=\frac{\exp(-u_{r+1})}{r!}\cdot u_{r+1}^r\cdot \prod_{j=1}^rj\cdot u_j^{j-1}\,,\,u_1,\ldots,u_r\in[0,1],\,u_{r+1}>0,
\end{align*}
which equals the product of the marginal densities $f_{U_{r+1}}(u_{r+1})=e^{-u_{r+1}}u_{r+1}^r/{r!}$, and $f_{U_j}(u_j)=j\cdot u_j^{j-1}$, $j\in\{1,\ldots,r\}$. 
\end{proof}
Since the right and left tail behavior of the normal distribution are symmetric and since the differences between subsequent extreme order statistics dominate the ones of intermediate order statistics, the auxiliary lemmas and Proposition \ref{order} suffice to conclude the main result.
\begin{theo}For $(X_j)_{1\le j\le n}$ i.i.d.\ $\mathcal{N}(0,1)$, it holds that
\begin{align}\label{dehtest}\lim_{n\to\infty}\P\Big(\sqrt{2\log(n)}\max_{2\le j\le n}\big(X_{(j)}-X_{(j-1)}\big)\le x\Big)=\prod_{j=1}^{\infty}\big(1-\exp(-j\cdot x)\big)^2\,.\end{align}
\end{theo}
The result combines \eqref{exptest} with the non-obvious, asymptotic independence of the differences. Note that the cdf of independent random variables equals the product of their cdfs. Since the differences are not identically distributed, the limit distribution does not belong to the class of standard extreme value distributions for maxima of i.i.d.\ sequences. Nevertheless, the limit cdf is remarkably simple and intuitive what I was not aware of before exploring this path. After (re-)discovering this result, I expected that it has been discussed in the literature and it was not difficult to find it as Thm.\ 1 in \cite{Deheuvels}. With the main focus on a related law of the iterated logarithm, \cite{Deheuvels} provides a rigorous, more technical and less intuitive proof of the convergence \eqref{dehtest}. I will therefore call the limit \emph{Deheuvels distribution}. The square on the right-hand side of \eqref{dehtest} is due to the symmetry of the tails. Looking only at one of the tails, we obtain the limit cdf without the square. This is useful when testing for positive and negative jumps separately. In order to compute quantiles based on \eqref{dehtest}, one can approximate the infinite product by a finite one up to some cut-off, or, even simpler, approximate it by $1-\exp(-x)-\exp(-2x)$, for $x$ not too small. This approximation exploits a telescoping sum and is very precise for all relevant quantiles. 

\begin{figure}[t]\begin{framed}
\includegraphics[width=4.45cm]{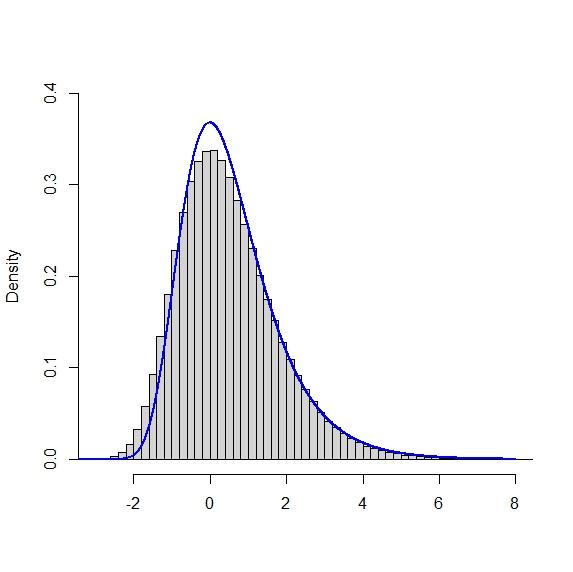}\includegraphics[width=4.45cm]{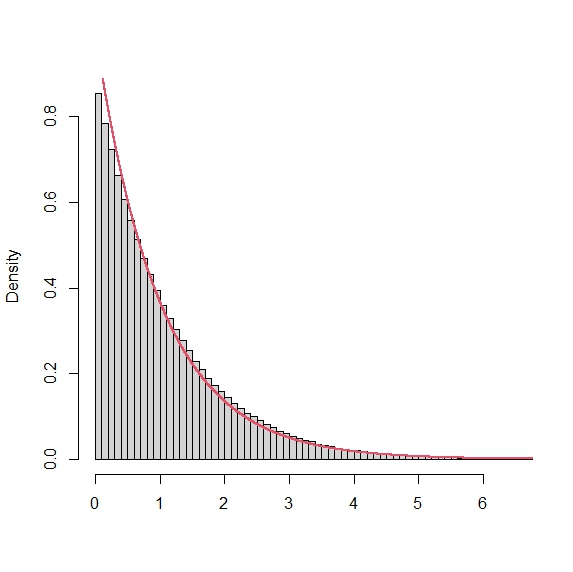}\includegraphics[width=4.45cm]{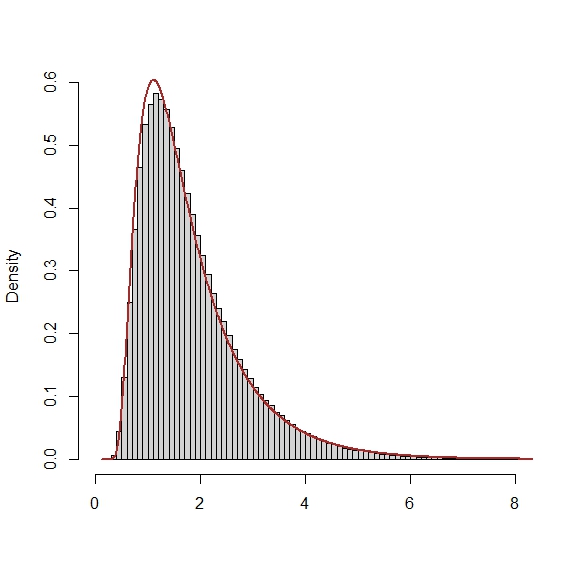}
\end{framed}
\caption{\label{Fig:Histo}Histograms of the statistics for $n=3{,}600$ from 1{,}000{,}000 Monte Carlo iterations and the densities of their asymptotic {\color{blue}{Gumbel}}, {\color{red}{exponential}} and {\color{darkbrown}{Deheuvels}} distributions.}
\end{figure}

Figure \ref{Fig:Histo} compares for $(X_j)_{1\le j\le n}$ i.i.d.\ $\mathcal{N}(0,1)$ histograms of the statistics
\begin{enumerate}
\item $\sqrt{2\log(n)}\cdot \big(X_{(n)}-\sqrt{2\log(n)}+\log\big(4\pi\log(n))/(2\sqrt{2\log(n)})\big)$ left-hand side,
\item $\sqrt{2\log(n)}\cdot \big(X_{(n)}-X_{(n-1)}\big)$ in the middle,
\item $\sqrt{2\log(n)}\cdot \max_{2\le j\le n}\big(X_{(j)}-X_{(j-1)}\big)$ right-hand side,
\end{enumerate}
for finite sample size $n=3{,}600$, corresponding to one price observation per second over one hour, based on a Monte Carlo simulation with 1{,}000{,}000 iterations, to the densities of the limit {\color{blue}{standard Gumbel}}, {\color{red}{standard exponential}} and {\color{darkbrown}{Deheuvels}} distributions. The derivative of the infinite product not having a nice closed form, I use a numerical approximation with Richardson's extrapolation to evaluate the density.

Crucial for the test is the precision of the fit in the high quantiles. We illustrate it based on our Monte Carlo simulation in Figure \ref{Fig:QQ} plotting empirical $(90+j)$\% percentiles, $0\le j\le 9$, against their theoretical asymptotic counterparts. As common in quantile-quantile (q-q) plots, we draw a diagonal line and the closer the points are to the diagonal, the better the fit by the limit distribution. We see that all three limit distributions fit the empirical, finite-sample distributions reasonably well. In fact, the fit for the differences of order statistics are better than that of the Gumbel distribution. I did, however, not try different variants of $(b_n)$ here which could further improve the Gumbel approximation, cf.\ \cite{nunes}. 
\begin{figure}[t]\begin{framed}
\includegraphics[width=4.45cm]{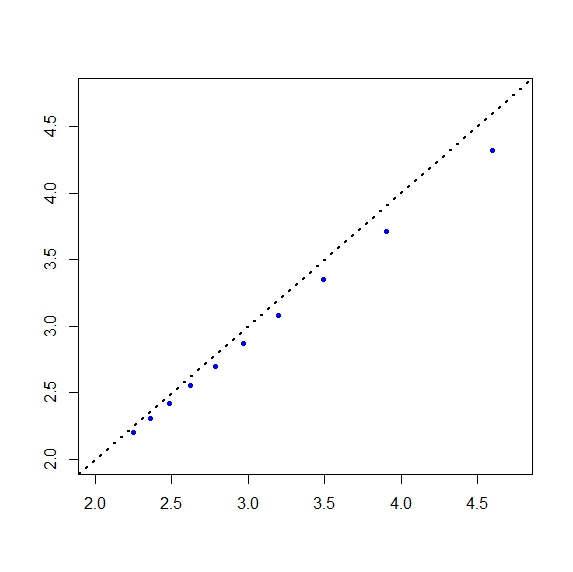}\includegraphics[width=4.45cm]{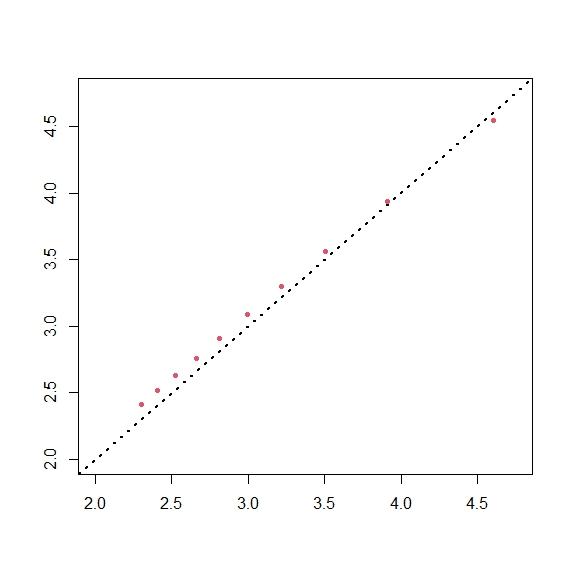}\includegraphics[width=4.45cm]{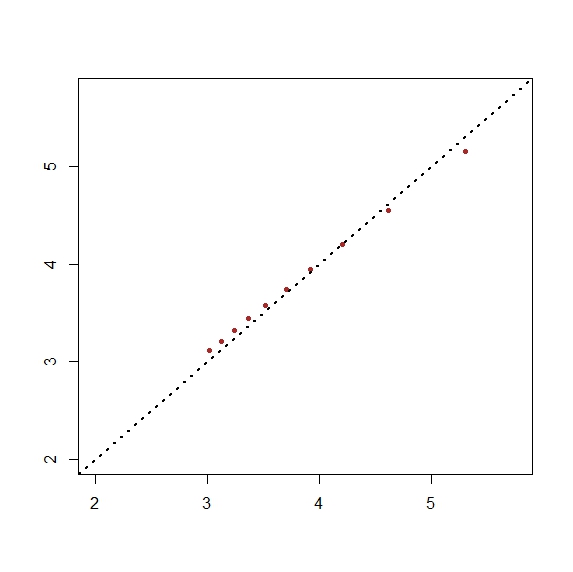}
\end{framed}
\caption{\label{Fig:QQ}Q-q plots with $(90+j)$\% percentiles, $0\le j\le 9$, of the statistics for $n=3{,}600$ from 1{,}000{,}000 Monte Carlo iterations compared to the asymptotic {\color{blue}{Gumbel}}, {\color{red}{exponential}} and {\color{darkbrown}{Deheuvels}} distributions.}
\end{figure}

We finish this section with our new \emph{R\'{e}nyi test} for jumps. Based on
\begin{align*}\sqrt{2\log(n)}\cdot \max_{2\le j\le n}\big(\Delta^n \hat X_{(j)}-\Delta^n \hat X_{(j-1)}\big) \todl \mathfrak{D}\,,\end{align*}
where $\mathfrak{D}$ is the Deheuvels distribution and \(\Delta^n \hat X=n^{1/2}(\Delta_1^n X/\hat\sigma_{1/n},\ldots, \Delta_n^n X/\hat\sigma_{1})\) the vector of normalized increments, we reject the null if the statistic left-hand side exceeds the $(1-\alpha)$ quantile of the Deheuvels distribution. The test has asymptotic level $\alpha$ and achieves the same rate of convergence as the Gumbel test.
 
In order to detect several jumps, the Gumbel test can be performed sequentially. In case of rejection, the time of the largest jump is estimated with the argmaximum. After discarding the largest absolute increment, the test is applied again. In case of another rejection, the next jump time is estimated. This is iterated until the test does not reject any more. For the R\'{e}nyi test, there is a similar sequential application. In case of rejection, however, we can readily ascribe all increments above or below the maximal difference of the order statistics to jumps. Since the maximum can be taken between several increments which contain jumps, we nevertheless apply another test which may be based on \eqref{exptest}.

\section{Is volatility rough?\label{sec:4}}
A fractional Brownian motion (fBm), $(B_t^H)_{t\ge 0}$, with \emph{Hurst exponent} $H\in(0,1)$, is a Gaussian process with continuous paths uniquely determined by $\E[B_t^H]=0$ for all $t$, and  
\[\E[B_t^H\,B_s^H]=\frac12 (t^{2H}+s^{2H}-|t-s|^{2H})\,,t,s\ge 0\,.\]
$(B_t^H)$ has stationary Gaussian increments $(B^H_t-B^H_s)\sim N(0,|t-s|^{2H})$, which are positively correlated for $H>1/2$, and negatively correlated for $H<1/2$. Except the case of a standard Brownian motion when $H=1/2$, increments are thus not independent and $(B_t^H)$ is not a Markov process and also not a semi-martingale. The fBm is self-similar with index $H$ given by the Hurst exponent, such that $a^{-H}B^H_{at}$ is distributed as $B_t^H$ for all $a>0$. Interested readers find a nice survey about fBm in \cite{nourdin}. Harold Edwin Hurst was in fact not a mathematician, but a British hydrologist who empirically found \emph{long-range dependence} in a time series of his measurements of the water level in the Nile river. Long-range dependence refers to a high degree of persistence in the data and after fBm was introduced by \cite{mandelbrot}, it can be modelled  by a fBm with large Hurst exponents. Such long memory was attributed in finance to volatility processes and \cite{comte} suggested a fractional Ornstein-Uhlenbeck process, with $H>1/2$, as a model for the \emph{log-}volatility. The Hurst exponent determines at the same time the regularity of the process in Assumption \ref{volass} and by the Kolmogorov-Chentsov continuity theorem the paths are Hölder continuous for any index strictly smaller than $H$. A recent strand of literature considers a \emph{rough fractional stochastic volatility} model built on the same kind of processes but with small Hurst exponents $H<1/2$. This development was initiated by \cite{roughvola} and is mainly motivated by empirical evidence. It is important to point out that related literature is looking at volatility processes over longer time periods and not on an intra-daily basis over, e.g., just one single day. The strategy of \cite{roughvola} is to consider a time series of realized volatilities based on high-frequency, intra-daily data over some longer period. Modelling integrated volatilities, or realized volatilities directly, by a fractional process, the latent volatility becomes observable, either directly or with negligible noise from the estimation. Based on $\sigma_{j\Delta},0\le j\le n$, they study the statistics
\begin{align}\label{mq}m(q,\Delta) = n^{-1}\sum_{k=1}^n\left|\log(\sigma_{k\Delta})-\log(\sigma_{(k-1)\Delta})\right|^q\,,q>0.\end{align}
\begin{figure}[t]
\includegraphics[width=7cm]{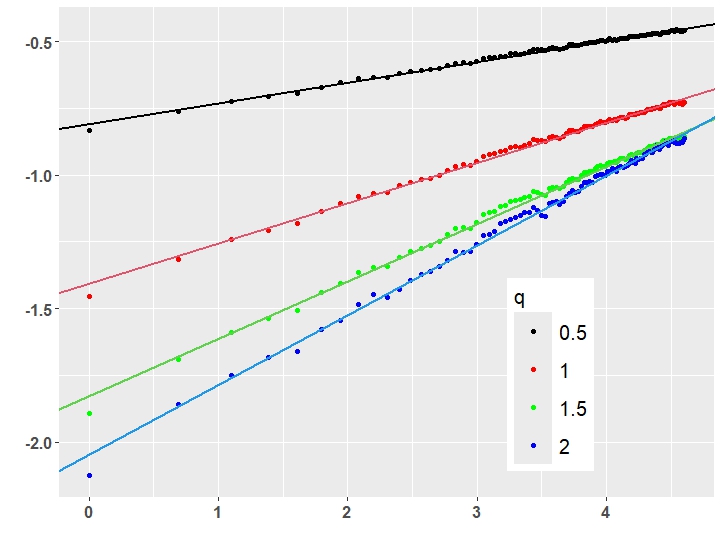}\includegraphics[width=7cm]{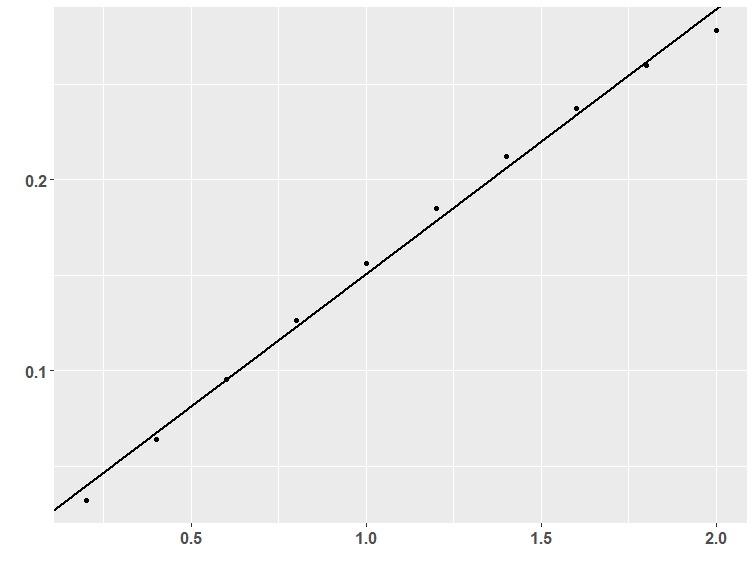}
\caption{\label{Fig:Rough}Left $m(q,\Delta)$ from \eqref{mq} as function of $\log(\Delta)$, right: $\zeta_q$ as a function of $q$.}
\end{figure}

The idea is to perform linear regressions what we motivate here differently than in \cite{roughvola}. Based on the defining properties of fBm above, we see for some time step $\Delta$ and $l,k\in\N$, $l\le k$, that 
\begin{align*}\log|\sigma\cdot B^H_{k\Delta}-\sigma\cdot B^H_{(k-l)\Delta}|&=\log(\sigma)+H\cdot \log(l\Delta)+\log|Z|\,,\end{align*}
with $Z\sim\mathcal{N}(0,1)$. This already resembles the model equation of a linear model, i.e., a linear function of $\log(l\Delta)$ with slope $H$. Having observations $\sigma_{j\Delta},0\le j\le n$, we compute \eqref{mq} over different coarser grids, or equivalently with $\log(\sigma_{k\Delta})-\log(\sigma_{(k-l)\Delta}), 1\le l\le L$, up to some $L\in\N$, and regress $m(q,l\Delta)$ on $\log(l\Delta)$ to estimate intercept and slope with a simple linear regression. If $(\log(\sigma_t))$ was a fBM, or as well if it was a more general fractional process, we expect to find $q\cdot H$ as the slope in these regressions. This and also more refined estimators of the Hurst exponent yield in several empirical studies of financial data similar results with Hurst exponents smaller than 0{.}2.

The data sets from the Oxford-Man Institute used for illustrations in \cite{roughvola} are unfortunately not available any more. We replicate the same behavior of statistics $m(q,\Delta)$ as in Figures 5-7 of \cite{roughvola} for a time series of 7021 quasi maximum likelihood daily volatility estimates from 1996 to 2023, based on the method by \cite{xiu}, inferred from intra-day high-frequency trade prices of the S\&P 500 market ETF. The data is constructed from the Risk Lab on Dacheng Xiu's website\,\footnote{\url{https://dachxiu.chicagobooth.edu/\#risklab}} and the S\&P 500 is certainly a very relevant financial index. It is not important if we insert $\sigma_{j\Delta}^2$, or a square root $\sigma_{j\Delta}$, in \eqref{mq}. The definition without square is taken from \cite{roughvola}, but we insert the estimates of squared volatility. The left plot in Figure \ref{Fig:Rough} illustrates the linear regressions for $q=j/2,1\le j\le 4$. The points give the computed statistics. The statistics $m(q,l\Delta)$, $1\le l\le 100$, look as a function of $l$ indeed almost perfectly logarithmic. This is confirmed by the good fit of the linear functions in the plot. The right-hand side of Figure \ref{Fig:Rough} compares the estimated slope, called $\zeta_q$ in \cite{roughvola}, along different values of $q$. From this illustration, we see the estimate $\hat H\approx 0{.16}$ for this data. Again, we find empirical evidence for a small Hurst exponent fitting a fBm to the log-volatilities. Moreover, our data shows pronounced negative empirical autocorrelations which further indicates small Hurst exponents and would contradict large ones.

This new rough volatility paradigm already stimulated a considerable body of research, beyond the high-frequency literature, for instance, on financial implications, in \cite{roughf1} and \cite{roughf}. The main motivation from econometrics to use this model is that it facilitates improved volatility forecasting, see \cite{forecast}, among others. Having a Gaussian process, optimal prediction is feasible and given by conditional expectation. While the puzzle of rough volatility vs.\ volatility persistence is now to a large extent -- but not yet fully -- understood, forecasting mainly exploits a correlation structure. From this point of view, large Hurst exponents and very small ones could both favour a similar good performance of prediction, while the opposite is the case for values close to 1/2. The application of rough volatility for forecasting uses the continuous-time model rather as a substitute of time series models over longer periods, where the latency of volatility is less crucial than within the framework of intra-daily high-frequency observations. The question if we can infer the Hurst exponent, or more general the regularity $\alpha$ from Assumption \ref{volass}, based on observations of the log-price $(X_{j\Delta_n})$ is nevertheless of great theoretical interest. Given its crucial role in spot volatility estimation in Section \ref{sec:2}, it is moreover practically relevant. 

It is known from \cite{Mathieu} that, based on $(n+1)$ high-frequency observations, the Hurst exponent $\alpha$ of the latent volatility can be estimated with optimal rate $n^{-1/(4\alpha+2)}$, if $\alpha>1/2$, exploiting results from \cite{gloterhoffmann}. The important question for rough volatility, if this is true also in case that $\alpha<1/2$, is confirmed in the recent work \cite{chong2}. Estimation methods and asymptotic confidence are furthermore established in the companion work \cite{chong1}. This is shown for models in that the log-volatility follows a fractional process of similar nature as fBm. The Hurst exponent $\alpha$ is in this case not only the regularity from Assumption \ref{volass}, but determines also the inter-dependence structure (persistence) and more. In joint work with Moritz Jirak, we are interested in the question, if the regularity $\alpha$ can be identified from high-frequency log-prices $(X_{j\Delta_n})$ in the more general case. Since for direct observations, the rates are the same, and most estimators for the Hurst exponent in this framework are in fact constructed to assess the regularity, this could be expected. However, we obtain a rather negative result with the following lower bound.  We impose regularity $\alpha$ in \eqref{volass2} and that the process exploits this regularity in the sense of a lower and an upper bound. It is clear that only the upper bound from Assumption \ref{volass} is not a suitable condition when we aim to estimate $\alpha$, since, e.g., constant functions satisfy this for any $\alpha$. 

\begin{theo}\label{theolower}
Suppose that positive constants $c_{\sigma}$ and $ C_{\sigma}$ exist, such that for $s,t\ge 0$:
\begin{align}\label{volass2}c_{\sigma} s^{2\alpha}\le \E\big[(\sigma_{(t+s)}-\sigma_{t})^2\big]  \le C_{\sigma} s^{2\alpha}\,.\end{align}
The minimax lower bound for estimation of $\alpha$ is determined by
\[ \exists \delta>0:\;\liminf_{n\to\infty}\inf_{\hat\alpha_n}\max_{\alpha\in\{\alpha_0,\alpha_0+\delta
r_n\}}
{\mathbb{P}_{\hspace*{-.1em}\alpha_0}}\big(|\hat\alpha_n-\alpha_0|\ge  \delta r_n\big)>0\,,
\]
with $r_n=(n^{-1/2+2\alpha_0})/\log(n)$. That is, for any sequence of estimators $\hat\alpha_n$ of the true parameter $\alpha_0$, $r_n$ gives a lower bound on the rate with that the minimax risk decreases in $n$.
\end{theo}
 
The proof is provided in Section \ref{sec:7}. Lower bounds for minimax rates typically rely on statistical groundwork by \cite{Tsybakov2008}. We exploit techniques and results from \cite{Tsybakov2008} for the proof of Theorem \ref{theolower} and our construction mimics one used in \cite{BJV} for a related, different lower bound pertaining change-points of $\alpha$. In our model, different from the one of \cite{chong2}, $\alpha$ determines only the regularity. The proof of the lower bound utilizes a sub-model which does not have the dependence structure of fBm. Since lower bounds extend to supersets but not to subsets, it does not apply to the more specific model with a fBm and is hence not in conflict with the result from \cite{chong2}. In particular, we obtain $n^{-1/2+2\alpha}$, for $\alpha<1/4$, as a lower bound. This shows that a consistent estimator only exists for $\alpha_0<1/4$! For $\alpha_0$ close to zero we get close to the standard parametric rate $n^{-1/2}$. Both rates from \cite{chong2} and Theorem \ref{theolower} have in common that the rate hinges on the parameter and is better for smaller values. The comparison reveals that estimation of a latent volatility's \emph{regularity}, or the \emph{Hurst exponent} imposing a model with a fractional process, are in general different problems. We conclude that estimating the regularity is statistically more difficult.

Let me finish the section with a positive result. We use the stochastic Landau symbols. Assume that $\hat\alpha$ is consistent
\begin{align}\hat\alpha-\alpha=\mathcal{O}_{\P}\big(n^{-1/2}\Delta_n^{-2\alpha}\big)\,,\end{align}
with $n^{-1/2}\Delta_n^{-2\alpha}\to 0$, such that $(\hat\alpha-\alpha)\stackrel{\P}{\longrightarrow} 0$, and consider the spot volatility estimator \eqref{spotvola} with optimal $k_n\propto\Delta_n^{-\frac{2\alpha}{2\alpha+1}}$. Not knowing $\alpha$, we replace it by $\hat\alpha$, and call the resulting estimator $\hat\sigma_s^{2,ad}$. The elementary identities
\begin{align*}\frac{\hat\alpha}{2\hat\alpha+1}-\frac{\alpha}{2\alpha+1}=\frac{\hat\alpha-\alpha}{(2\alpha+1)(2\hat\alpha+1)}\,,\end{align*}
and
\begin{align*}\Delta_n^{\frac{\hat\alpha-\alpha}{(2\alpha+1)(2\hat\alpha+1)}}&=\exp\Big(\frac{\alpha-\hat\alpha}{(2\alpha+1)(2\hat\alpha+1)}\log\Delta_n^{-1}\Big)\\
&=1+\frac{\alpha-\hat\alpha}{(2\alpha+1)(2\hat\alpha+1)}\log\Delta_n^{-1}+\mathcal{O}_{\P}\Big(n^{-1}\Delta_n^{-4\alpha}(\log\Delta_n^{-1})^2\Big)\\
&=1+\KLEINO_{\P}(1)\,,\end{align*}
then yield with the results from Section \ref{sec:3} that for two random variables $Z_1$ and $Z_2$:
\begin{align*}
\hat\sigma_s^{2}-\hat\sigma_s^{2,ad}&=\sigma_s^2+Z_1 \,\Delta_n^{\frac{\alpha}{2\alpha+1}}-\sigma_s^2-Z_2\,\Delta_n^{\frac{\hat\alpha}{2\hat\alpha+1}}\\
&=\Delta_n^{\frac{\alpha}{2\alpha+1}}(Z_1-Z_2+\KLEINO_{\P}(1))=\mathcal{O}_{\P}\big(\Delta^{\frac{\alpha}{2\alpha+1}}\big)\,.
\end{align*}
We conclude that $\hat\sigma_s^{2,ad}$ attains the same optimal rate of convergence as the estimator which exploits known $\alpha$.
\section{Limit order microstructure noise\label{sec:5}}
\begin{figure}[t]
\begin{framed}
\hspace*{-.35cm}\includegraphics[width=14cm]{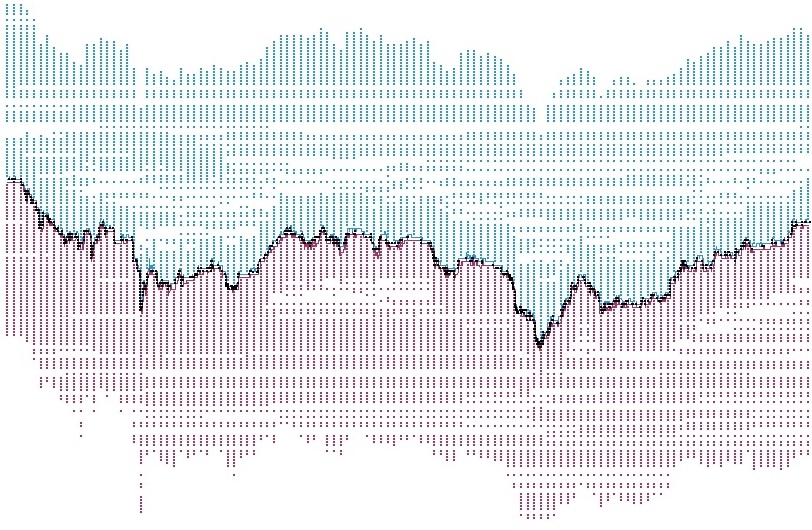}
\end{framed}
\caption{\label{Fig:LOMN}{\color{violet}{Bid}} and {\color{light-blue}{ask}} quotes and trade prices (black dots) for the Apple asset over a 10 minutes time interval.}
\end{figure}

While we modelled high-frequency log-prices so far as discretizations of continuous-time stochastic processes, when having available data from a limit order book, there is not only one price at some given time. Figure \ref{Fig:LOMN} gives a snapshot of price dynamics of the Apple asset traded at Nasdaq over a 10 minutes time interval. We use Nasdaq data from Lobster.\,\footnote{\url{lobsterdata.com}} A blue line shows the evolution of the best ask price, that is, the lowest price at which someone is willing to sell the asset. A red line shows the best bid price, that is, the highest price someone is offering to buy the asset. In between there is a bid-ask spread. The many points above the best ask illustrate many other active ask-limit orders and below the best bid active bid-limit orders. Trading usually takes places when market orders arrive with that someone buys or sells the asset at the best available price. These are executed against the available limit orders. For this reason trade prices bounce between the best ask and best bid what makes the illustration of all three in the same plot a bit overfraught. Trade prices are plotted in Figure \ref{Fig:LOMN} as black dots. A prevalent concept for market microstructure in financial econometrics is to assume some underlying efficient, semimartingale log-price process in an arbitrage-free market modelling longer-term price dynamics, while high-frequency observations are diluted by an additive market microstructure noise. Therefore, the observation model to account for market microstructure is
\begin{align}\label{noise}Y_i=X_{t_i}+\epsilon_i~,~0\le i\le n,\end{align}
with an It\^{o} semimartingale $(X_t)$ and noise $(\epsilon_i)$. Such a model was proposed in \cite{abdl00}, among others, for trade prices with regular noise and there is a vast area of research on this model. Classical regular {\textbf{M}}arket {\textbf{M}}icrostructure {\textbf{N}}oise ({\textbf{MMN}}) $(\epsilon_i)_{0\le i\le n}$ is i.i.d.\ with $\E[\epsilon_i]=0$. If a full limit order book is available, it is recently applied to mid quotes, i.e., averages of best bid and best ask quotes. If we model the prices of (best) ask quotes directly, a natural assumption is that they all lie above the efficient, semimartingale log-price $(X_t)$. Reasons are that ask orders will typically be submitted at prices above the level that is seen as current fair price to make money and they also lie above the trade prices. This leads us to a \emph{stochastic boundary model} with observations in the epigraph of a semimartingale boundary process. We hence use model \eqref{noise} with {\textbf{L}}imit {\textbf{O}}rder {\textbf{M}}icrostructure {\textbf{N}}oise ({\textbf{LOMN}}) which satisfies
\begin{align}\epsilon_i\stackrel{i.i.d.}{\sim}F_{\eta},\,\epsilon_i\ge 0\,,\end{align}
that is, {\textbf{L}}ower-bounded, {\textbf{O}}ne-sided {\textbf{M}}icrostructure {\textbf{N}}oise. The model was introduced in \cite{BJR}. We assume that $(\epsilon_i)_{0\le i\le n}$ is exogenous with a cdf 
\begin{align}\label{cdf}F_{\eta}(x) =\eta x\,\big(1+\KLEINO(1)\big) ,\;\mbox{as}~x\downarrow 0\,.\end{align}
Bid prices are analogously modelled with noise that is upper bounded and both combined in practice. Although Figure \ref{Fig:LOMN} shows prices in a discrete image space under a very high time resolution, it is standard to work with the real-valued process $(X_t)$, to perform estimation of the volatility, or other daily quantities. It is then natural to consider continuous noise distributions also. Since our methods use differences between local minima or maxima of the data only, it is not crucial that the boundary of the noise is exactly zero. It can be some unknown constant instead, or even a regular function over time, what is meaningful to include compensation of market processing costs. This possible generalization is one reason why we model ask and bid prices separately in boundary models, e.g., instead of considering noise on a bounded interval. Moreover, a model with noise on an interval would not simplify the statistical problem but rather complicate the situation. Condition \eqref{cdf} does not impose a parametric form of the noise. The assumed standard behaviour of the cdf close to the boundary is satisfied by many common distributions, as a uniform distribution on some interval $[0,A]$, $A>0$, an exponential distribution as we know from Section \ref{sec:3}, and a heavy-tailed (shifted) Pareto distribution. Nevertheless, we currently work on generalizations of the model to allow for some general tail index which is 1 in \eqref{cdf}. The irregular, non-negative noise leads to statistical inference based on {\emph{local minima}} instead of local averages which are used under regular noise in the literature. This is motivated by the problem of estimating boundary parameters in parametric statistics. We explain the key idea looking at the prominent example of the taxi problem. An important advantage of LOMN and using order statistics compared to MMN is that no conditions on the right tail of the noise distribution or on the existence of moments of the noise are required. 

In our stochastic boundary model we do of course not have a constant boundary to estimate as in the taxi problem, but want to recover a latent semimartingale boundary process. This situation is intricate, but -- although the approach appears venturous -- we approximate the boundary process locally constant over small time blocks. From the analogy to the taxi problem, it is then natural to estimate the efficient log-price locally by local block-wise minima
\begin{align*}m_{k,n}=\min_{i\in\mathcal{I}_k^n}Y_i\,,~\mathcal{I}_k^n=\{t_i^n: ~t_i^n \in (kh_n,(k+1)h_n)\},\,0\le k\le h_n^{-1}-1\,.
\end{align*}
Let us assume for simplicity equidistant observations again, $t_i^n=i/n$, and $h_n^{-1}\in\N$ being the sequence of number of blocks, and $nh_n\in\N$ the number of observations per block. In our asymptotic high-frequency regime, $h_n\to 0$, and $nh_n\to\infty$, as $n\to\infty$. There is a balanced regime, $h_n\propto n^{-2/3}$, in which the stochastic order of the minimal error over a block and the movement of the boundary process over a block are the same, as
\begin{align*}\min_{i\in\mathcal{I}_k^n}\epsilon_i=\mathcal{O}_{\P}\big( (nh_n)^{-1}\big),~\text{and}~\big(X_{(k+1)h_n}-X_{kh_n}\big)=\mathcal{O}_{\P}\big(h_n^{1/2}\big)\,.\end{align*}
Based on local minima in this balanced regime, a rate-optimal estimator of the integrated volatility has been established in \cite{BJR}. 

\begin{figure}[H]
{\colorbox{light-blue}{\parbox{0.98\linewidth}{
\textbf{Taxi problem}\\[-.3cm]

\begin{minipage}[l]{0.67\textwidth}
Imagine you go for a walk in New York City and notice a lot of the famous yellow cab taxis. You're wondering how many there are in total. 
Fortunately, the yellow cabs are labeled with consecutive integers on their engine covers. So you can note the numbers you see during your walk and then estimate the unknown maximal number based on your sample.\,\footnotemark[6]
\end{minipage}
\begin{minipage}[r]{0.28\textwidth}
\hspace*{.5cm}\includegraphics[width=3cm]{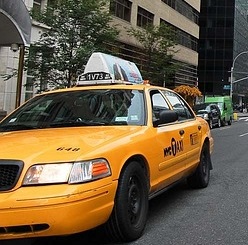}
\end{minipage}
\vspace*{.2cm}

Consider the similar problem of estimating the upper boundary $\theta$ from i.i.d.\ uniformly $U([0,\theta])$-distributed random variables $X_1,\ldots, X_n$ on the interval $[0,\theta]$. A very successful (though not in this example) and convenient construction for a point estimator is the method of moments: Since $X_1$ has expectation $\theta/2$, and the sample average $\bar x_n$ is a good estimator for an expectation, set $\hat\theta^{MM}=2\bar x_n$. This estimator is unbiased, converges almost surely to $\theta$, and satisfies a central limit theorem 
\[\sqrt{n}\big(\hat\theta^{MM}-\theta\big)\todl \mathcal{N}(0,\theta^2/3)\,,\] 
based on which asymptotic confidence intervals can be obtained. Looking at the likelihood \( L(\theta;x_1,\ldots,x_n)=\theta^{-n}{\1}\{\theta\in[x_{(n)},\infty)\}\), i.e., the product density as a function of $\theta$, however, tells statisticians that the maximum $X_{(n)}$ is a sufficient statistic. That means $X_{(n)}$ preserves all information about $\theta$ we have from $X_1,\ldots,X_n$. Therefore, during your walk you do not need to take notes with all observed numbers, but only remember the largest number that you observed. Based on the likelihood, we obtain the maximum likelihood estimator: $\hat\theta^{ML}=X_{(n)}$. We find that $n(\theta-X_{(n)})\todl \text{Exp}(\theta^{-1})$, since
\[\P\big(n(\theta-X_{(n)}\big)> x)=\P\big(X_{(n)}<\theta-x/n\big)=\big(1-x/(n\theta)\big)^n\to \exp(-x/\theta)\,.\]
Once more in this article we get an exponential limit distribution! The rate is much faster than for $\hat\theta^{MM}$, but $\hat\theta^{ML}<\theta$ is obviously biased.\\
Since $X_{(n)}$ is not only sufficient, but moreover complete, statisticians know that the associated Rao-Blackwell improvement of some unbiased $L_2$-estimator yields the unbiased estimator with uniformly smallest variance (umvu) by Lehmann-Scheff\'e. Since $\hat\theta^{MM}$ is unbiased, let us determine its Rao-Blackwell improvement: 
\[ \hat\theta=\E_\cdot[\hat\theta^{MM}\,|\,X_{(n)}]=\frac{2}{n}\sum_{i=1}^n\E_\cdot[X_i\,|\,X_{(n)}]=2\,\E_\cdot[X_1\,|\,X_{(n)}]\,.\,\footnotemark[7] 
\]
With probability $1/n$ we have $X_1=X_{(n)}$. Otherwise, the conditional distribution of $X_1$ given $X_{(n)}=x_{(n)}$ is uniform: $U([0,x_{(n)}])$. We hence obtain the umvu estimator \(\hat\theta=2(1/n+(n-1)/(2n))X_{(n)}=(n+1)X_{(n)}/n\).}}}
\end{figure}
\footnotetext[6]{The same problem was highly relevant in world war \RNum{2}, since the Germans stamped serial numbers on their tanks. From the observations on tanks that were captured or broke and left behind, British statisticians estimated the production rate of German tanks very precisely. Outside Germany the problem is since then known as the \href{https://en.wikipedia.org/wiki/German_tank_problem}{German tank problem.}}
\footnotetext[7]{The dot emphasizes that $\E_\cdot[\,-\,\,|\,X_{(n)}]$ does not depend on $\theta$.}

Like the estimators based on the maximum in the taxi problem converge faster than with the standard rate, \cite{BJR} proved that their estimator attains an optimal rate $n^{-1/3}$, with that the root mean squared error tends to zero, which improves upon the well-known standard rate $n^{-1/4}$ for regular noise. Lower bounds for the rate and the asymptotic variance under regular noise in the parametric case were established by \cite{gloter}. 
The distribution of local minima in the balanced regime is, however, involved which yet limited available results. In particular, in \cite{BJR} we could not provide asymptotic confidence for the integrated volatility. The article \cite{bibinger2022} contributes a step forward in this direction and extends the probabilistic theory required to work with the boundary model. For the tail function of local minima, we conclude with conditioning, \eqref{noise}, \eqref{cdf}, a Taylor expansion and dominated convergence that
\begin{align*}
&\P\Big(h_n^{-1/2}\big(m_{k,n}-X_{kh_n}\big)\Big)>x\sigma_{kh_n}\Big)\\
~&=\E\bigg[\exp\Big(\sum_{i=knh_n+1}^{(k+1)nh_n}\log\big(1-F_{\eta}\big(h_n^{1/2}\sigma_{kh_n}\big(x-h_n^{-1/2}(W_{i/n}-W_{kh_n})\big)\big)\big)\Big)\bigg]\\
~&=\E\Big[\exp\Big(-nh_n^{3/2}\sigma_{kh_n}\eta\int_0^1(B_t-x)_{-}\,\text{d}t\,(1+\KLEINO(1))\Big)\Big]
\end{align*}
for all $x<0$, with a standard Brownian motion $(B_t)$. To work with the integrated negative part of a Brownian motion in the last expression, we exploit and extend results about \emph{local time} of Brownian motion. One main ingredient of the asymptotic analysis in \cite{bibinger2022} is an expansion of this tail function based on a generalized arcsine law. Here, we focus on a simpler idea which is nevertheless the most important step to approximate the distribution of the local minima. Selecting blocks slightly larger than in the balanced regime, we have $nh_n^{3/2} \to\infty$ in the exponent, such that the probability tends to zero unless the integral yields zero. This is the case if and only if the event $\{\min_{0\le t\le 1} B_t\ge x\}$ occurs. In this regime, we hence obtain that
\begin{align*}
&\P\Big(h_n^{-1/2}\big(m_{k,n}-X_{kh_n}\big)\Big)>x\sigma_{kh_n}\Big)\\
~&=\P\big(\min_{0\le t\le 1} B_t\ge x\big)+\KLEINO(1)\,.
\end{align*}
The distribution of the minimum of a Brownian motion over the interval $[0,1]$ is remarkably simple. This is due to the \emph{reflection principle} connected with the strong Markov property of $(B_t)$. We derive with the reflection principle from the above approximation that for $x<0$, since $\P\big(\min_{0\le t\le 1} B_t\ge x\big)=\P\big(|B_1|\ge -x\big)$, that 
\[-h_n^{-1/2}\big(m_{k,n}-X_{kh_n}\big)\stackrel{st}{\longrightarrow} H\negthinspace M \negthinspace N(0,\sigma_{kh_n}^2)\,,\]
as $nh_n^{3/2}\to \infty$. The distribution of $|Z|$, for $Z\sim\mathcal{N}(0,1)$, is called \emph{half-normal distribution}. Since our limit is distributed as the product $\sigma_{kh_n}|Z|$ then, we call it mixed half-normal.

For volatility estimation, with $(B_t)$ and $(\tilde B_t)$ two independent standard Brownian motions, define
\begin{align}\label{psi}\Psi_{\negthinspace n}(\sigma^2)=h_n^{-1} \E\Big[\Big(\min_{i\in\{0,\ldots,nh_n-1\}}\big(\sigma B_{\frac{i}{n}}+\epsilon_i\big)-\min_{i\in\{1,\ldots,nh_n\}}\big(\sigma \tilde B_{\frac{i}{n}}+\epsilon_i\big)\Big)^2\Big].\end{align}

\begin{figure}[H]
{\colorbox{light-blue}{\parbox{0.98\linewidth}{
\textbf{Reflection principle}\\[-.5cm]

\begin{minipage}[l]{0.5\textwidth}
Based on the symmetry of the normal distribution and the independent normal increments of Brownian motion, many nice and explicit formulas for functionals of Brownian motion are available. An important tool to deduce such formulas is the following symmetry argument known as reflection principle.
\end{minipage}
\begin{minipage}[r]{0.41\textwidth}
\hspace*{.5cm}\includegraphics[width=4.75cm]{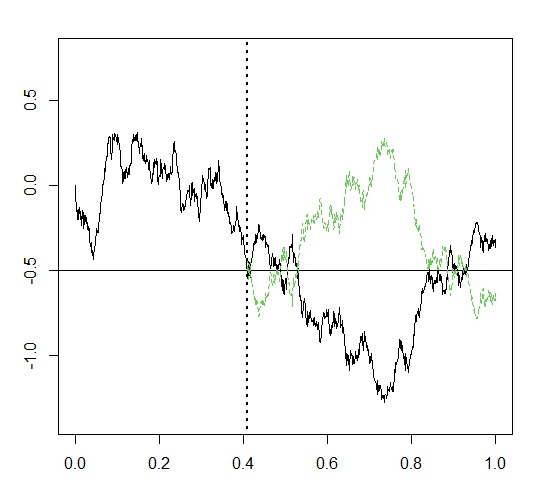}
\end{minipage}
\vspace*{.2cm}

Denote with $T_x$ the first entry time of the Brownian motion in some $x<0$, which is a \emph{stopping time}. For $x=-0{.}5$, it is sketched by the vertical dashed line in the plot. We define the \emph{reflected process} as
\[B_t^R=\begin{cases}B_t,&0\le t\le T_x\\ 2B_{T_x}-B_t=2x-B_t,& t\ge T_x\end{cases}\,.\]
In the plot, we see a simulated path of $(B_t)$ drawn as a solid line. The reflected process shares the same path up to time $T_x$, while after time $T_x$ it takes the reflected path given by the dashed green line. By the strong Markov property the process $(B_{T_x+s}-B_{T_x})_{s\ge 0}$ is a standard Brownian motion which is independent from the past, such that the \emph{reflected process is again a standard Brownian motion}. This is very useful, since whether or not the pathwise minimum moves below some level $x<0$ is determined from the terminal values $B_1$ and $B_1^R$: 
\begin{align*}\P\big(\min_{0\le t\le 1} B_t\le x\big)&=\P\big(\min_{0\le t\le 1} B_t\le x,B_1>x)+\P\big(B_1\le x\big)\\
&=\P(B_1^{R}\le x)+\P\big(B_1\le x\big)=2\,\P\big(B_1\le x\big)=2\,\Phi(x)\,,
\end{align*}
with $\Phi$ the cdf of the standard normal distribution of $B_1$. We exploit the reflection principle once more to even access the \emph{joint distribution of terminal value and minimum}. For $x\le 0$, and $x\le y$, we have
\begin{align*}
\P\big(\min_{0\le t\le 1} B_t\le x,B_1\le y\big)&=\P\big(\min_{0\le t\le 1} B_t\le x\big)-\P\big(\min_{0\le t\le 1} B_t\le x,B_1\ge y\big)\\
&=2\Phi(x)-\P\big(\min_{0\le t\le 1} B_t\le x,B_1\ge y\big)\\
&=2\Phi(x)-\P\big(T_x\le 1,B_1\ge y\big)\\
&=2\Phi(x)-\P\big(T_x\le 1,B_1-B_{T_x}\ge y-x\big)\\
&=2\Phi(x)-\P\big(T_x\le 1,B_1^R-B_{T_x}^R\ge x-y\big)\\
&=2\Phi(x)-\P\big(T_x\le 1,B_1^R\ge 2x-y\big)\\
&=2\Phi(x)-\P\big(T_x\le 1,B_1\ge 2x-y\big)\\
&=2\Phi(x)-\Phi(2x-y)\,,
\end{align*}
where we conclude the last line since $(2x-y)\le x$.
}}}
\end{figure}

\enlargethispage*{1cm}
From the moments of the half-normal distribution, as $nh_n^{3/2}\to \infty$, we obtain that
\begin{align}\Psi_{\negthinspace n}(\sigma^2)=\frac{2(\pi-2)}{\pi}\sigma^2+\KLEINO(1)\,.\end{align}
Not having any lower bound for the integrated negative part in the above expression, the remainder decays however slowly in $n$, and we require an asymptotic expansion and a numerical approximation of $\Psi_{\negthinspace n}$ for an estimator with desirable properties.
Nevertheless, with $K_n\to\infty$, we first consider the simple estimator 
\begin{align}\hat\sigma^2_{\tau}=\frac{\pi}{2(\pi-2)K_n}\sum_{k=(\lfloor h_n^{-1}\tau\rfloor-K_n)\vee 1}^{\lfloor h_n^{-1}\tau\rfloor-1}h_n^{-1}\big(m_{k,n}-m_{k-1,n})^2\,,\end{align}
in case without price jumps and a truncated version which is robust to nuisance jumps. A main result of \cite{bibinger2022} is that under Assumption \ref{volass} for $K_n=C_K h_n^{\delta -2\alpha/(1+2\alpha)}$, with some constants $C_K$ and $\delta$, $0<\delta<2\alpha/(1+2\alpha)$, the estimator satisfies the stable clt
\begin{align}K_n^{1/2}\Big(\hat\sigma^2_{\tau}-\frac{\pi}{2(\pi-2)}\Psi_{\negthinspace n}\big(\sigma_{\tau}^2\big)\Big) \stackrel{st}{\longrightarrow} \mathcal{N}\Big(0,\frac{7\pi^2/4-2\pi/3-12}{(\pi-2)^2}\sigma^4_{\tau}\Big)\,.\end{align}
In fact, we use the (approximated) function $\Psi_{\negthinspace n}$ for a bias correction to obtain a clt at optimal rate. The asymptotic variance is derived with the expansion of the tail function based on the joint distribution of minimum and terminal value of a Brownian motion over $[0,1]$ concluded from the reflection principle. To this end we use one of the most important examples for applications of Fubini-Tonelli in probability that relates moments and the tail function: For some non-negative random variable $Z$, with distribution $\P_Z$, and $k\in\N$, it holds true that
\begin{align*}
\E[Z^k]=\int_0^{\infty}z^k\,\text{d}\P_Z(z)&=\int_0^{\infty}\Big(\int_0^z ky^{k-1}\text{d}y\Big)\,\text{d}\P_Z(z)\\
&=k\int_0^{\infty}y^{k-1}\int_y^{\infty}\text{d}\P_Z(z)\,\text{d}y=k\int_0^{\infty}y^{k-1}\P(Z>y)\text{d}y\,.
\end{align*}
Integration with respect to the $\sigma$-finite probability and Lebesgue measures is exchanged here. Extensions to covariances and real-valued random variables are available and allow us to use the form of the tail function from above.

In a recent preprint \cite{bibinger2024jump}, we develop jump detection methods under LOMN including a Gumbel test for jumps. It is based on
\[T^{BHR}=\max_{k=1,\ldots,h_n^{-1}-1}\Big|\frac{m_{k,n}-m_{k-1,n}}{\big(\hat\sigma_{kh_n}^{2}\big)^{1/2}}\Big|.\]
Based on extreme value theory, we show that under the null hypothesis 
\[H_0:\sup_{\tau\in [0,1]}|\Delta X_{\tau}|=0\] 
and for $(\sigma_t)\in C^{\alpha}$, i.e., Hölder continuous with regularity $\alpha$, it holds with $h_n=2\log(2h_n^{-1}-2)n^{-2/3}$ that
\begin{align}\label{testgumbel}n^{1/3}\,T^{BHR}-2\log(2h_n^{-1}-2)+\log\big(\pi\log(2h_n^{-1}-2)\big)\todl\Lambda\,,\end{align}
with $\Lambda$ the standard Gumbel distribution. Under local alternatives

\[H_1:\liminf_{n\to\infty}{\dbl{n^{\beta}}}\sup_{\tau\in(0,1)}|\Delta X_{\tau}|>0,~\mbox{for some}~{\dbl{\beta<1/3}}\,,\]
the test satisfies with $q_{1-a}^{\thinspace\Lambda}=-\log\left(-\log\left(1-a\right)\right)$ that
\begin{align}\lim_{n\to\infty}\P_{H_1}\Big(n^{1/3}\,T^{BHR}-B_n>q_{1-a}^{\thinspace\Lambda}\Big)=1\,.\end{align}
The subscript of the measure is to indicate that we are under $H_1$, and the path has at least one jump. Considering local alternatives, the question about which probability space(s) to work on is justified, but not particularly important here, since we can simply consider the distributions of the statistics directly to avoid an arduous construction.

The main insight of this result is that under LOMN {\emph{smaller jumps can be identified}} compared to MMN. While we can detect jumps of size larger than $n^{-1/3}$, only jumps of size larger than $n^{-1/4}$ can be found under MMN. Moreover, working with order statistics to infer jumps has some nice advantages compared to local averages under MMN, where averaging over jump times is creating huge problems described as ``pulverisation of jumps by pre-averages'' by \cite{zhangmykland3}. This is illustrated in Section 2 of \cite{bibinger2024jump}. One ingredient to show \eqref{testgumbel} is uniform consistency of the spot volatility estimation, for which we require the continuity of $(\sigma_t)$ under $H_0$. Furthermore, the precise Gumbel convergence for differences between half-normal random variables is determined, since we cannot trace this one back to a standard example of extreme value theory. Our sequence is furthermore not i.i.d., but it is known that Gumbel convergences of maxima hold analogously more generally under weak dependence conditions. 

\section{Outlook\label{sec:6}}
In the multi-dimensional framework with a portfolio of $d$ stocks, the volatility process becomes $(d\times d)$ matrix-valued. The key role for risk diversification is rather taken by the covariances than the idiosyncratic volatilities. A \emph{co-jump} pattern is of interest to separate idiosyncratic and systemic effects, see e.g., \cite{caporin}. Since the estimation uncertainty of a $(d\times d)$ matrix increases proportional to $d^4$ in the dimension, we have our very own curse of dimensionality. 

Multivariate ultra high-frequency data are not only subject to market microstructure, but discrete observations moreover arrive at \emph{non-synchronous} times. Volatility matrix estimation under these peculiarities motivated another strand of research. In \cite{BHMR} we contributed two main insights:
\begin{enumerate}
\item Different than for non-noisy observations, non-synchronicity effects are at first order asymptotically negligible. In a combination with noise, the noise prevails.
\item A lower bound for the asymptotic variance-covariance structure of volatility matrix estimation reveals that the multivariate model allows improved estimates, also of idiosyncratic volatilities.
\end{enumerate}
The first result is based on an asymptotic equivalence between a continuous-time and the discrete-time observation model. Asymptotically equivalent experiments provide the same amount of information about unknown quantities, which hence can be estimated with the same precision in both situations. If one model is simpler than another one or already well explored, this is very useful, also since statistical methods can be transferred. The effect of efficiency gains from a multivariate model for the estimation of a single volatility arises when assets are correlated and observed with uncorrelated noise. 

The picture on boundary estimation sketched by the taxi problem is yet incomplete, since the rate of convergence heavily depends on the behaviour of the cdf close to the boundary. For instance, for a triangle distribution which is the convolution square of the uniform distribution, the rate is only $\sqrt{n}$ instead of $n$. Extending the model with a general tail index and its estimation allow to better calibrate boundary models to limit order quotes. Our first empirical trials indicate that different assets might show different tail behaviours, what is particularly interesting in view of their strong correlations and since a small tail index results in a higher accuracy of volatility estimation. This is a strong motivation to develop a multivariate observation model with limit order microstructure noise. When the estimation uncertainty varies across different stocks, a risk analysis for one stock, e.g., Apple, could be improved using data also from another stock, e.g., Google. Compared to multivariate regular noise, efficiency gains become even more relevant affecting the rates of convergence and not only minimal asymptotic variances at optimal rate.

Forecasts of financial risk can improve considerably when going from a model for a single stock price to a multivariate model, e.g., this was the case for the multivariate GARCH model proposed in \cite{bollerslev1988capital}. Consequently, if rough volatility provides accurate forecasts it should be further extended to a multivariate model. This should include possibly different Hurst exponents, what is mathematically challenging. Therefore, it is also of theoretical interest for mathematicians.

Currently, the analysis of \emph{high-dimensional high-frequency data} is a vibrant research area. This refers to an asymptotic regime in that not only $n\to\infty$, but moreover $d\to\infty$ is considered for an asymptotic expansion. In this area, high-frequency statistics is combined with methods from high-dimensional statistics, e.g., LASSO, penalization in general, shrinkage estimation, thresholding eigenvalues, principal component analysis and sparsity, see e.g., \cite{principal}, \cite{pelger}, \cite{chen2020five}, \cite{ledoit} and \cite{christensen}. In view of strong correlations between most financial assets, {\emph{factor models}} appear to be very attractive. These are of the form  
\[dX_t=B_t^q\,dF_t+dZ_t~,[F,Z]\equiv 0,~\Sigma_t=B_t^q\Sigma_t^{S}(B_t^q)^{\top}+\Sigma_t^{I}\,,\]
where the $q$ factors $F_t$ affect all stocks, with $B_t^q\in\R^{d\times q}, \Sigma^{S}_t\in\R^{q\times q}$. The dimension $q$ is kept fix as $d\to\infty$. The estimation of all components of the model is challenging. Moreover, the rank $q$ has a crucial role and we are interested in testing constant rank and detecting changes of $q$ over time. The precision matrix, the inverse of the (integrated) volatility matrix, is the most important object for optimal portfolio allocation. \cite{cai2020high} focusses on its estimation from high-dimensional high-frequency data. Assuming equidistant observations with regular microstructure noise, however, there is something left to improve upon in future research.

\section{Proof of Theorem 2\label{sec:7}}
It suffices to prove that $r_n$ is a lower bound for a specific sub-model contained in our general model, since the lower bound then extends to the general model. A simplified sub-model that preserves the main structure of the estimation problem is the estimation of $\alpha\in(0,1]$ from observations
\begin{align}\label{obs_lb}Y_j=\big(1+\Delta^{\alpha}U_j\big)Z_j~,~Z_j\stackrel{iid}{\sim}\mathcal{N}(0,1)\,,1\le j\le n=\Delta^{-1} T,\end{align}
where $(U_j)$ are i.i.d.\ real-valued random variables with a symmetric centered law, $\E[U_1^{2k-1}]=0,k\in\N$, independent of $(Z_j)_{j\ge 1}$, for which all moments exist. Assuming that the law of $U_1$ has a Lebesgue density $g_U$, we obtain by conditioning the following density of $\Pa$ with respect to the Lebesgue measure $\lambda$:
\begin{align*}
 \frac{\text{d}\Pa}{\text{d}\lambda}(x)\hspace*{-0.05cm}&=\int_{\R}\frac{1}{\sqrt{2\pi}}\frac{1}{1+\Delta^{\alpha}u}\exp\Big(\frac{-x^2}{2(1+\Delta^{\alpha}u)^2}\Big)g_U(u)\,\text{d}u  \\
&=\frac{1}{\sqrt{2\pi}}\int_{\R}\exp\Big(-\log\big(1+\Delta^{\alpha}u\big)-\frac{x^2}{2(1+\Delta^{\alpha}u)^2}\Big)g_U(u)\,\text{d}u \\
 &=\frac{1}{\sqrt{2\pi}}\int_{\R}\exp\Big(\sum_{k=1}^{\infty}(-1)^k\frac{(\Delta^{\alpha}u)^k}{k}-x^2\sum_{k=0}^{\infty}(-1)^k{(\Delta^{\alpha}u)^k}\frac{k+1}{2}\Big)g_U(u)\,\text{d}u \\
 &=\frac{1}{\sqrt{2\pi}}e^{-\frac{x^2}{2}}\int_{\R}\exp\Big(\sum_{k=1}^{\infty}(-1)^k(\Delta^{\alpha}u)^k\Big(\frac{1}{k}-\frac{k+1}{2}x^2\Big)\Big)g_U(u)\,\text{d}u \\
&=\frac{1}{\sqrt{2\pi}}e^{-\frac{x^2}{2}}\Big(1+\sum_{k\ge 1}C_{2k}(x)(\Delta^{\alpha})^{2k}\Big)\,,
\end{align*}
where addends for odd $k$ vanish by the symmetry of the law of $U_1$, and the coefficients of the power series are degree $2k$ polynomials in $x$, with
\begin{align*}
C_2(x)&=\E[U_1^4]\Big(\frac{x^4}{2}-\frac{5}{2}x^2+1\Big)\,,\\
C_4(x)&=\E[U_1^8]\Big(\frac{x^8}{24}-\frac{11}{12}x^6+\frac{41}{8}x^4-7x^2+1\Big)\,.
\end{align*}
Naturally, the first addend of $\text{d}\Pa/\text{d}\lambda(x)$ yields the standard normal density. We see that it holds that
\[\int_{\R}\frac{1}{\sqrt{2\pi}}e^{-\frac{x^2}{2}}\sum_{k\ge 1}C_{2k}(x)\,\text{d}x=0\,.\]
Note that $\int_{\R}C_{2k}(x)e^{-\frac{x^2}{2}}\,\text{d}x=0$, for the coefficients $k=2,4$, can be seen inserting the moments of the standard normal distribution. For two parameters $\alpha,\tilde{\alpha}\in(0,1]$, consider the $\chi^2$-divergence
\begin{align*}
&\chi^2\big(\text{d}\Pa\|\text{d}\Paa)=\int\Big(\frac{\text{d}\Pa}{\text{d}\Paa}-1\Big)^2\text{d}\Paa=\int\Bigg(\frac{\frac{\text{d}\Pa}{\text{d}\lambda}-\frac{\text{d}\Paa}{\text{d}\lambda}}{\frac{\text{d}\Paa}{\text{d}\lambda}}\Bigg)^2 \frac{\text{d}\Paa}{\text{d}\lambda}\,\text{d}\lambda\\
&=\int_{\R}\bigg(\frac{\sum_{k\ge 1}C_{2k}(x)\big((\Delta^{\alpha})^{2k}\hspace*{-0.05cm}-\hspace*{-0.05cm}(\Delta^{\tilde\alpha})^{2k}\big)}{1\hspace*{-0.05cm}+\hspace*{-0.05cm}\sum_{k\ge 1}C_{2k}(x)(\Delta^{\tilde\alpha})^{2k}}\bigg)^2\frac{e^{-\frac{x^2}{2}}}{\sqrt{2\pi}}\Big(1\hspace*{-0.05cm}+\hspace*{-0.05cm}\sum_{k\ge 1}C_{2k}(x)(\Delta^{\tilde\alpha})^{2k}\Big)\hspace*{0.025cm}\text{d}x\,,
\end{align*}
which is one common measure of the distance between the probability measures $\Pa$ and $\Paa$. $\chi^2\big(\text{d}\Pa\|\text{d}\Paa)$ tends to zero when $\Delta\to 0$. In a high-frequency asymptotic regime, $\Delta\to0$, the last equation yields that
\begin{align*}\chi^2\big(\text{d}\Pa\|\text{d}\Paa)&=\Delta^{4\alpha}\big(1-\Delta^{2(\tilde{\alpha}-\alpha)}\big)^2\int_{\R}(C_2(x))^2 \frac{e^{-\frac{x^2}{2}}}{\sqrt{2\pi}}\,\text{d}x\;\big(1+\KLEINO(1)\big)\\
&=\frac{13}{2}\Delta^{4\alpha}\big(1-\Delta^{2(\tilde{\alpha}-\alpha)}\big)^2\big(1+\KLEINO(1)\big)\,.
\end{align*}
By \cite[Lemma 2.7]{Tsybakov2008} the Kullback-Leibler divergence ${\bf D}_{KL}(\Pa\|\Paa)$ is bounded by $\chi^2\big(\text{d}\Pa\|\text{d}\Paa)$, and by additivity of the Kullback-Leibler divergence for product measures, we obtain that
\begin{align}{\bf D}_{KL}\Big(\Pa^{\otimes n}\|\Paa^{\otimes n}\Big)\le n\,\chi^2\big(\text{d}\Pa\|\text{d}\Paa)\,.\end{align}
This yields for a high-frequency asymptotic regime that
\[{\bf D}_{KL}\Big(\Pa^{\otimes n}\|\Paa^{\otimes n}\Big)\le n\frac{13}{2}\Delta^{4\alpha}\big(1-\Delta^{2(\tilde{\alpha}-\alpha)}\big)^2\big(1+\KLEINO(1)\big)\,.\]
Considering for $\tilde{\alpha}$ the sequence $\tilde{\alpha}^{(n)}=\alpha+\delta r_n$, with a constant $\delta$ and a null sequence $(r_n)$, we obtain that
\begin{align*}{\bf D}_{KL}\Big(\Pa^{\otimes n}\|\mathbb{P}_{\hspace*{-.1em}\tilde\alpha^{(n)}}^{\otimes n}\Big)&\le n\frac{13}{2}\Delta^{4\alpha}\big(1-\exp\big(-2\delta r_n\log(\Delta^{-1})\big)\big)^2\big(1+\KLEINO(1)\big)\\
&=n\Delta^{4\alpha}\,26\delta^2 r_n^2\big(\log(\Delta^{-1})\big)^2\big(1+\KLEINO(1)\big)\,.\end{align*}
Setting $r_n=n^{-1/2}\Delta^{-2\alpha}/\log(\Delta^{-1})$, the Kullback-Leibler divergence is bounded by the finite constant $26\delta^2$, and we conclude by \cite[Theorem 2.2]{Tsybakov2008}. \qed

\bibliographystyle{chicago}
\bibliography{literature}
\end{document}